\documentclass[a4paper,11pt]{article}
\usepackage{amssymb,amsthm}
\usepackage[table]{xcolor}
\usepackage{color}
\usepackage{graphicx,subfigure,epstopdf}
\usepackage{wrapfig}
\usepackage{amsfonts}
\usepackage{multirow}
\usepackage{dsfont}
\usepackage{caption}
\usepackage{soul}
\usepackage{colortbl}
\definecolor{bubbles}{rgb}{0.91, 1.0, 1.0}
\definecolor{maroon}{cmyk}{0,0.87,0.68,0.32}
\setulcolor{red}

\usepackage[colorlinks=true,linkcolor=blue,citecolor=red]{hyperref}
\usepackage{mathrsfs}
\usepackage{verbatim}
\usepackage[misc,geometry]{ifsym}
\newtheorem{obsv}{Observation}
\newtheorem{fact}{Fact}
\newtheorem{defn}{Definition}
\newtheorem{lema}{Lemma}
\newtheorem{theo}{Theorem}

\newtheorem{prop}{Property}
\newtheorem{clam}{Claim}
\usepackage[framemethod=tikz]{mdframed}
\newmdenv[backgroundcolor=magenta!5,
topline=false,
bottomline=false,
rightline=false,
skipabove=\topsep,
skipbelow=\topsep
]{siderules}
\usepackage{tcolorbox}
\usepackage{geometry}
\newcommand{\colb}[1]{{\color{blue}{{\it #1}}}}

\newcommand{\rbac}{RBAC}
\newcommand{\grbac}{GRBAC}

\newcommand{\rbrac}{RBRAC}
\newcommand{\rbcac}{RBCAC}
\newcommand{\grbrac}{GRBRAC}
\newcommand{\grbcac}{GRBCAC}
\newcommand{\cla}{${\cal A}$}

\newcommand{\clp}{${\cal P}$}

\title{Asymmetric Separation Problem for Bichromatic Point Set\footnote{A preliminary version of this paper appeared in the International Joint Conference on Theoretical Computer Science – Frontier of Algorithmic Wisdom (IJTCS-FAW) 2023.}}

\author{Sukanya Maji$^1$~\thanks{corresponding author} \and Supantha Pandit$^2$ \and Sanjib Sadhu$^1$}
\date{$^1$National Institute of Technology Durgapur,
India
\\ $^2$ Dhirubhai Ambani Institute of Information and Communication Technology, Gandhinagar, Gujarat, India}

\begin{document}
\maketitle
\begin{abstract}
We study the Generalized Red-Blue Annulus Cover problem for two sets of points, red ($R$) and blue ($B$), where each point $p \in R\cup B$ is associated with a positive penalty \clp$(p)$. The red points have non-covering penalties, and the blue points have covering penalties. The objective is to compute an annulus (either a rectangular or a circular) \cla~such that the value of the function \clp$({R}^{out}) +$\clp$({ B}^{in})$ is minimum, where ${R}^{out} \subseteq {R}$ is the set of red points not covered by \cla~and  ${B}^{in} \subseteq {B}$ is the set of blue points covered by \cla. We study the problem for various types of axis-parallel rectangular annulus and circular annulus in one and two dimensions. We also study a restricted version of the rectangular annulus cover problem, where the center of the annulus is constrained to lie on a given horizontal line $L$. We design a polynomial-time algorithm for each type of annulus.
\end{abstract}


\renewcommand*{\thefootnote}{\fnsymbol{footnote}}

%
%
%
%\linenumbers
\section{Introduction}
 In this paper, we study the asymmetric separation problem~\cite{armaselu2024} for a set of bichromatic points $P=R\cup B$, where $R$ and $B$ are the set of $n$ red points and $m$ blue points, respectively. In this problem, we compute a geometric annulus of  rectangular or circular shaped as the asymmetric separator for $R\cup B$ so that the annulus covers all the red points in $R$ and the minimum number of blue points in $B$. An annulus is a closed region bounded by two geometric objects of the same type, such as axis-parallel rectangles \cite{GluchshenkoHT09}, circles \cite{RoyZ92}, 
%triangles \cite{}, 
convex polygons \cite{BarequetG14}, and many more. The minimum width annulus problem is related to the problem of finding a shape that fits in a given set of points in the plane \cite{Abellanas03}.
Further, the largest empty annulus of different shapes has many potential applications \cite{books/sp/PreparataS85,BergCKO08}.
A point $p$ is said to be covered by an annulus ${\cal A}$ if $p$ lies within the closed region bounded by the two boundaries of ${\cal A}$.

In this paper, we study the Generalized Red-Blue Annulus Cover problem to compute an
%One of the generalizations of ``the problem of computing a minimum width annulus'' is computing an 
annulus for a bichromatic point-set where each point has a positive penalty. In this problem, we take a bichromatic point-set as a set of red points and a set of blue points. A positive penalty, also called the non-covering (resp. covering) penalty is associated with each red (resp. blue) point, i.e., if a red point is not covered, then its penalty is counted, whereas if a blue point is covered, then its penalty is counted. The objective is to find an annulus $\cal A$ with minimum weight, where the weight of $\cal A$ is the total penalty of the red points not covered plus the total penalty of the blue points covered by $\cal A$. We define this problem formally as follows.
\begin{siderules}
    \textbf{\color{red!70!blue}{Generalized Red-Blue Annulus Cover (\grbac) problem. }} Given a set 
    of $n$ red points ${R}=\{q_1,q_2,\ldots,q_n\}$ and a set 
    of $m$ blue points ${B}=\{p_1,p_2,\ldots,p_m\}$, and a penalty function ${\cal P}: {R}\cup {B} \rightarrow \mathbb{R}^+$. The objective is to find an annulus (either rectangular or circular) $\cal A$ such that the quantity $\lambda = \sum_{q \in {R}^{out}} {\cal P}(q) + \sum_{p \in {B}^{in}} {\cal P}(p)$ is minimum, where ${R}^{out} \subseteq {R}$ is the set of points not covered by $\cal A$ and  ${B}^{in} \subseteq {B}$ is the set of points covered by $\cal A$. Note that $\lambda$ gives the measure of penalty of ${\cal A}$.
\end{siderules}

If the penalty of each red point is {\boldmath $\infty$} and that of each blue point is a {\bf unit}, then the \grbac~problem reduces to a problem of computing an annulus that covers all the red points and a minimum number of blue points. We refer to this problem as the \colb{Red-Blue Annulus Cover (\rbac)} problem.

We focus our study on computing either a rectangular or a circular shape annulus in both the \rbac~and \grbac~problems. We refer to the problem of computing a rectangular annulus in the \rbac~problem (resp. \grbac~problem) as the \rbrac~problem (resp. \grbrac~problem), whereas computing a circular annulus in the \rbac~problem (resp. \grbac~problem) is referred to as the \rbcac~problem (resp. \grbcac~problem).
There exist multiple annuli covering the same set of red and blue points that minimize $\lambda$.
 Throughout this paper, we report the one among these annuli, whose boundaries pass through some red points as an optimal solution. 
\newline\newline
{\bf Motivation:}
The oncologists remove the cancer cells while avoiding the healthy cells as much as possible during the surgery or radiation therapy of the cancer patients. We distinguish these two types of cells with two different colors. The excision of the tumor leads us to various geometric optimization problems, e.g., to determine the smallest circle enclosing the red points or that separates the red points from the blue points. 
Instead of computing the smallest enclosing circle, the computation of an annulus which covers all the tumor cells and a minimum number of healthy cells, leads us to less wastage of the healthy cells while performing surgery.
This motivates us to formulate a problem of covering the bichromatic point-set by an annulus that covers all the given red points while minimizing the number of blue points covered by it. Depending on the stage of the cancer cells and the importance of the healthy cells, we can assign different weights called penalties to each cancer cell and healthy cell so that the penalty incurred due to the elimination of healthy cells and keeping the cancer cells within the body during the surgery, is minimized. This motivates us to study the annulus cover problem by assigning penalties to each point in $R\cup B$.  

\section{Related Works}
\label{reldwrk}
%Finding the largest empty annuli of different shapes in a given point-set is well-studied. This problem has many potential applications \cite{books/sp/PreparataS85,BergCKO08}. Díaz-Báñez et al. \cite{Diaz-BanezHMRS02} first studied the problem for circular empty annulus and proposed an $O(n^3\log n)$-time and $O(n)$-space algorithm to solve it. 
Armaselu~\cite{armaselu2024} introduced the asymmetric separation problem, where the separator partitions the bichromatic point set into two subsets, one containing all the red points with the minimum number of blue points, and the other containing the remaining blue points.
They computed an axis-parallel maximum-area separating rectangle in $O(m \log m + n)$ time and $O(m + n)$ space, where $m$ and $n$ are the numbers of blue and red points, respectively. In addition, they considered the planar arbitrary orientation version and presented an algorithm that takes $O(m^3 + n \log n)$ time and $O(m + n)$ space.

The problem of computing the largest empty annulus~\cite{books/sp/PreparataS85,BergCKO08} is a well-studied problem, where a set $S$ of $n$ points is partitioned into two subsets of points - one lying completely outside the outer boundary and the other one lying completely inside the inner boundary of the annulus, i.e., no point in $S$ lies in the interior of the annulus. This problem is related to the maximum facility location problem for $n$ points, where the facility is a circumference. Díaz-Báñez et al.~\cite{Diaz-BanezHMRS02} computed the largest empty circular annulus in $O(n^3\log n)$-time and $O(n)$-space.
Bae et al.~\cite{BaeBM21} studied the maximum-width empty annulus problem, which computes square and rectangular annulus in $O(n^3$) and $O(n^2\log n)$ time, respectively. Both of these algorithms require $O(n)$ space. Also, Paul et al.~\cite{PaulS020} presented a combinatorial technique using two balanced binary trees that computes the maximum empty axis-parallel rectangular annulus that runs in $O(n \log n)$ time.
Abellanas et al. \cite{Abellanas03} presented a linear time algorithm
for the rectangular annulus problem such that the annulus covers a given set of points. They considered several variations of this problem. Gluchshenko et al. \cite{GluchshenkoHT09} presented an $O(n \log n)$-time optimal algorithm for the planar rectilinear square annulus. Also, it is known that a minimum width annulus in arbitrary orientations can be computed in $O(n^3 \log n)$ time for the square annulus \cite{Bae18} and in $O(n^2 \log n)$ time for the rectangular annulus \cite{MukherjeeMKD13}.

Barbay et al.~\cite{ipl/BarbayCNP14} presented a quadratic time algorithm for the {\it Maximum-Weight Box} problem that computes an axis-parallel box $D$ with maximum weight $W(D)$ for a given set $P$ of $n$ weighted points (the weight of each point $p\in P$ is either positive or negative) where $W(D)$ is given by the sum of the weights of the points lying inside $D$. Dobkin et al.~\cite{jcss/DobkinGM96} solved the {\it Maximum Bichromatic Discrepancy Box} problem, where a box that maximizes the absolute difference between the number of Red and Blue points it contains is computed, in $O(n^2\log n)$ time, $n$ being the total number of red and blue points. Eckstein et al.~\cite{coap/EcksteinHLNS02} has proved that the {\it Maximum Box} problem which computes a box containing the maximum number of blue points and no red points, is NP-hard if $d$, the dimension of the points is part of the input.  This problem in two dimensions was solved in $O((m+n)^2\log (m+n))$ time by Liu and Nediak~\cite{cccg/LiuN03}, and later it was improved by Backer et 
al.~\cite{latin/BackerK10} which needs $O((m+n)\log^3 (m+n))$ time. 
Bereg et al.~\cite{iccsa/BeregDZR15} studied the problem of computing a circle of the smallest radius such that the total weight of the points covered by the circle is maximized. This needs $O(m^2(m+n)\log (m+n))$ time and $O(m+n)$ space. Edelsbrunner and Preparata~\cite{EDELSBRUNNER1988218} computed a convex polygon with the fewest edges that separates two sets of $n$ points in the plane if it exists in $O(n \log n)$ time.\\
Abidha and Ashok~\cite{Abidha2020} studied the geometric separability problem for a bichromatic point-set $P = R \cup  B$ ($|P|=n$) of red ($R$) and blue ($B$) points. It computes a separator as an object of a certain type that separates $R$ and $B$. They computed different types of separators, e.g., (i) a non-uniform and uniform rectangular annulus of fixed orientation in $O(n)$ and $O(n \log n)$ time, (ii) a non-uniform rectangular annulus of arbitrary orientation in $O(n^2 \log n)$ time, (iii) a square annulus of fixed orientation in $O(n \log n)$ time, (iv) an orthogonal convex polygon in $O(n \log n)$ time. Bitner et al.~\cite{bitner2010minimum} computed the largest separating circle and the minimum separating circle for bichromatic point-set in $O(m (n + m) log (n + m))$ time and $O(mn)$ + $O^{*}(m^{1.5})$ time, respectively, where the $O^{*}()$ notation ignores some polylogarithmic factor.

The geometric version of a red-blue set cover problem is NP-Hard~\cite{carr1999red}, where for a given set of bichromatic (red and blue) points and a set of objects, a subset of objects is chosen that cover all the blue points and the minimum number of red points. Chan and Hu~\cite{CHAN2015380} designed PTAS for a weighted geometric set cover problem where the objects are 2D unit squares. Madireddy et al.~\cite{madireddy2021geometric} provided the APX-Hardness results of various special red-blue set cover problems. Shanjani~\cite{shanjani2020hardness} also showed that the Red-Blue Geometric Set Cover is APX-hard when the objects are axis-aligned rectangles. 
Abidha and Ashok~\cite{abidha2022red} studied the parameterized complexity of the generalized red-blue set cover problem.
Bereg et al.~\cite{bereg2012class} studied the class cover problem for axis-parallel rectangles and presented a $O(1)$-approximation algorithm.
\\

{\bf Outline of the paper}: We introduce the essential concepts, terminologies, and notations followed by our contributions in Section~\ref{prelim}. In this paper, we first study the one-dimensional versions of the $RBRAC$ and $GRBRAC$ problems in Sections~\ref{joco1} and~\ref{joco2} before proceeding to the corresponding problems in the two-dimensional versions, which are discussed in Sections~\ref{joco3} and~\ref{joco4}. In a one-dimensional problem, the red and blue points lie on the real line $L$. Section~\ref {joco5} describes the restricted version of the $GRBRAC$ problem in two dimensions. Then we discuss the circular annulus cover problem in Section~\ref{joco6}, and its two variations, the $RBCAC$ problem and $GRBCAC$ problems in Sections~\ref{joco7} and~\ref{joco8}.

\section{Preliminaries and Notations}
\label{prelim}
Throughout this paper, we study the axis-parallel rectangular annulus and circular annulus for the bichromatic point-set, and we use rectangular annulus to imply axis-parallel rectangular annulus unless otherwise stated. For any point $p$, we denote its $x$ and $y$-coordinates by $x(p)$ and $y(p)$, respectively. 

{\bf \it Notations used for rectangular annulus:} A rectangular annulus $\cal A$ is the closed region bounded by two parallel rectangles, one lying completely inside the other, and we denote these inner and outer rectangles of $\cal A$ by ${\cal R}_{in}$ and ${\cal R}_{out}$, respectively. 
We use $left({\cal R})$, $right({\cal R})$, $top({\cal R})$ and $bot({\cal R})$ to denote 
the left, right, top, and bottom sides of a rectangle $\cal R$, respectively. 
\begin{figure}[h]
    \centering
    \includegraphics[width=0.7\textwidth]{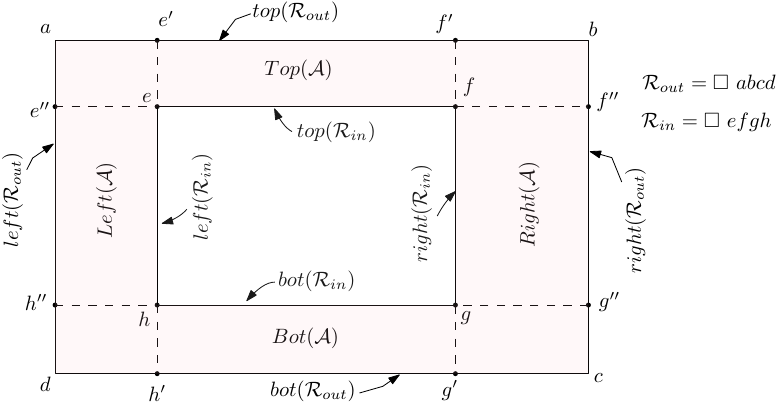}
    \caption{Annulus $\cal A$ with its different attributes.}
    \label{Attributes}
\end{figure}
We use $\Box abcd$ to represent an axis parallel rectangle with its four vertices at $a$, $b$, $c$, and $d$. A rectangular annulus ${\cal A}$ has four regions viz., $Left(\cal A)$, $Right(\cal A)$, $Top(\cal A)$ and $Bot(\cal A)$. The left sided region $Left(\cal A)$ is the closed rectangular region bounded by $left({\cal R}_{out})$, $top({\cal R}_{out})$, $bot({\cal R}_{out})$ and the line passing through $left({\cal R}_{in})$, e.g., $\Box ae'h'd$ in the Figure~\ref{Attributes}. Similarly, the other regions $Right(\cal A)$, $Top(\cal A)$ and $Bot(\cal A)$ of the annulus $\cal A$ are defined which are shown as $\Box f'bcg'$, $\Box abf''e''$ and $\Box h''g''cd$, respectively in the Figure~\ref{Attributes}.

The  two rectangles ${\cal R}_{out}$ and ${\cal R}_{in}$ of a rectangular annulus $\cal A$ define its four widths: $w_\ell$ (left width), $w_r$ (right width), $w_t$ (top width), and $w_b$ (bottom width). The horizontal distance between the left (resp. right) boundaries of ${\cal R}_{in}$ and ${\cal R}_{out}$ defines the left width $w_\ell$ (resp. right width $w_r$). Similarly, the vertical distance between the top (resp. bottom) sides of ${\cal R}_{in}$ and ${\cal R}_{out}$ defines the top width $w_t$ (resp. bottom width $w_b$). 
We say a rectangular annulus $\cal A$ to be \colb{uniform} if all of its four widths are the same i.e., $w_\ell = w_r = w_t$ = $w_b$ (see Figure~\ref{fig:uni}). Otherwise, we say that $\cal A$ is \colb{non-uniform} (see Figure~\ref{fig:non-uniform-non-con},~\ref{fig:non-uni}). It is observed that in a uniform annulus $\cal A$, its two rectangles ${\cal R}_{in}$ and ${\cal R}_{out}$ must be concentric. Further, in a \colb{non-uniform concentric annulus},  $w_\ell = w_r$  and $w_t = w_b$ (see Figure~\ref{fig:non-uni}), whereas all the widths are different in a \colb{non-uniform non-concentric annulus}. In the {\color{blue} non-uniform concentric annulus} $\cal A$, it's left (resp. top) width $w_\ell$ (resp. $w_t$) and right (resp. bottom) width $w_r$ (resp. $w_b$) are the same, and we refer to this width as the {\color{blue} horizontal} (resp. {\color{blue} vertical}) {\color{blue} width}, which is denoted by $w_h$ (resp. $w_v$).\\
\begin{comment}
\begin{figure}[h]
    \centering
    \includegraphics[width=1.0\textwidth]{fig.pdf}
    \caption{(a) Uniform annulus, (b) Non-uniform annulus, and (c) Non-uniform concentric annulus. Points lying inside the ${\cal R}_{in}$ depict the centers of the rectangles.}
    \label{fig:types-of-annuli}
\end{figure}
\end{comment} 
\begin{figure}[t] 
 \captionsetup{belowskip=-5pt}
\begin{center}
\subfigure[ ]{\includegraphics[width=.3\textwidth]{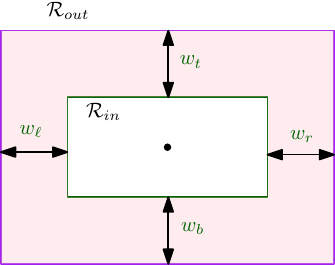}
\label{fig:uni}
}
\hfill
\subfigure[ ]{\includegraphics[width=.3\textwidth]{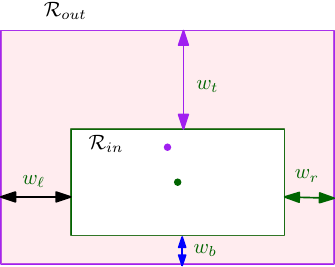}
\label{fig:non-uniform-non-con}
}
\hfill
\subfigure[ ]{\includegraphics[width=.3\textwidth]{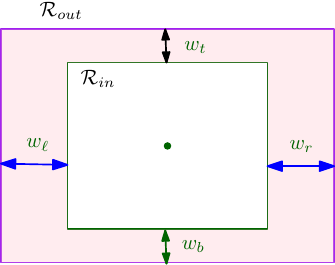}
\label{fig:non-uni}
}
\end{center}
\vspace{-0.15in}
\caption{%Different types of rectangular annulus, 
(a) Uniform annulus, (b) Non-uniform annulus, and (c) Non-uniform concentric annulus. Points lying inside the ${\cal R}_{in}$ depict the centers of the rectangles.}
\label{fig:types-of-annuli}
\end{figure}

{\bf \it Notations used for circular annulus:} 
A circular annulus is made up of two concentric circles of different radii.
The width of a circular annulus $\cal A$ is the difference between the radii of its outer and inner circles, which we denote by ${\cal C}_{out}(\cal A)$ and ${\cal C}_{in}(\cal A)$, respectively.
We use $VD(R)$ and $FVD(R)$ to denote the Voronoi diagram and the farthest-point Voronoi diagram of the point-set $R$, respectively. 
We use $VD_i$ and $VD_{i,j}$ to denote the  Voronoi diagrams for the point-set $R\cup\{p_i\}$ and $R\cup\{p_i,p_j\}$ where $p_i,p_j\in B$, respectively. So, $VD_{i}=VD(R\cup\{p_i\})$ and $VD_{i,j}=VD(R\cup\{p_i,p_j\})$. 
Similarly, $FVD_i$ and $FVD_{i,j}$ are used to denote the farthest-point Voronoi diagrams for the point-set $R\cup\{p_i\}$ and $R\cup\{p_i,p_j\}$ where $p_i,p_j\in B$, respectively i.e., $FVD_{i}=FVD(R\cup\{p_i\})$ and $FVD_{i,j}=FVD(R\cup\{p_i,p_j\})$.\\

{\bf Our Contributions:} To the best of our knowledge, there exists no works that are exactly similar to the problems studied in this paper.  Table \ref{chap4tab1} shows our contributions. We need to mention that, each solution needs $O((m+n)\log(m+n))$ preprocessing time and $O(m+n)$ space, where $m$ and $n$ are the numbers of blue and red points, respectively.\\

\begingroup
\renewcommand{\arraystretch}{1.2} % Default value: 1
\begin{table}[ht]
\centering 
\begin{tabular}{|l|c|c|l|l|}
\hline
\rowcolor{bubbles}
{\bf Separator} & \multicolumn{1}{l|}{\bf Input} & {\bf Type} & {\bf Annulus} & {\bf Time complexity} \\ \hline
\hline
\multirow{10}{*}{Rectangular}        & \multirow{4}{*}{$(R\cup B)$ on ${\mathbb R}^1$}   & \multirow{2}{*}{$RBAC$}  & Non-uniform   & $O(m+n)$ \\ \cline{4-5} 
 \multirow{10}{*}{Annulus}      &                       &                    & Uniform       & $O(m+n)$     \\ \cline{3-5} 
      &              & \multirow{2}{*}{$GRBAC$}    & Non-uniform  & $O(m+n)$       \\ \cline{4-5} 
       &              &                     & Uniform  & $O(n^2(m+n))$        \\
\cline{2-5} 
  & & \multirow{5}{*}{$RBAC$}    & Non-uniform  & $O(n(m+n))$  \\ 
 &                        &          & Non-Concentric  & \\  \cline{4-5} 

      &         &                   & Non-uniform   & $O(n(m+n)^2)$    \\ 
 &                        &    &Concentric         & \\ \cline{4-5} 

     & \multirow{2}{*}{$(R\cup B)$ on ${\mathbb R}^2$}              &               & Uniform   & $O(n(m+n)^2)$        \\  \cline{3-5} 
 &         &       & Non-uniform   & $O(n^4(m+n)^2)$ \\ 
 &                        & \multirow{3}{*}{$GRBAC$}         &  Non-Concentric  & \\  \cline{4-5} 

     &      &                   & Non-uniform  & $O(n^4(m+n)^2)$    \\
 &                        &            & Concentric   & \\  \cline{4-5} 

       &               &          & Uniform   & $O(n^2(m+n)^3)$   \\ \cline{1-3} \cline{4-5} 
\multirow{3}{*}{Rectangular}
         &                                &        &  Non-uniform  & $O(n^5(m+n))$        \\ 
 &                        &     \multirow{3}{*}{$GRBAC$}    &  Non-Concentric  & \\  \cline{4-5} 

\multirow{2}{*}{Annulus centered}        &   \multirow{1}{*}{$(R\cup B)$ on ${\mathbb R}^2$}                             &          &  Non-uniform  & $O(n^3(m+n)^2)$        \\  
 &                        &         &  Concentric  & \\  \cline{4-5} 

\multirow{1}{*}{on a given line}         &                                &        & Uniform   & $O(n^2(m+n)^2)$        \\  \cline{1-3} \cline{3-5}
\multirow{1}{*}{Circular}     &   \multirow{2}{*}{$(R\cup B)$ on ${\mathbb R}^2$}          &     $RBAC$                       & Uniform  &     $O(m^2n(m+n))$        \\ 
 \cline{3-5}
\multirow{1}{*}{Annulus}                &   & $GRBAC$   & Uniform   &$O(n^2(m+n)^3)$\\\hline
\end{tabular}
\vspace{.2cm}
\caption{Results obtained for asymmetric separation of $R\cup B$ with annulus.}
\label{chap4tab1}
%\vspace{-0.65cm}
\end{table}
\endgroup

%\newpage
\subsection{Remarks}
\label{remarks}
%The computation of an optimal annulus $\cal A$ in $GRBAC$ problem depends not only on the penalty of the red (resp. blue) points lying outside (resp. inside) $\cal A$ but also on how these points are distributed over the plane.
There are studies related to the $GRBAC$ problem in the literature, e.g., {\it Maximum-Weight Box} problem~\cite{ipl/BarbayCNP14}, and the problem of computing a circle with maximum weight~\cite{iccsa/BeregDZR15} (see Section~\ref{reldwrk}). However, those studies are on rectangles and circles instead of annuli. The result of the problem on the rectangle (resp. circle) cannot imply the solution to our $GRBRAC$ (resp. $GRBCAC$) problem.
%due to the following contradictory example. 
%Consider the optimal rectangle ${\cal R}_1$ in~\cite{ipl/BarbayCNP14} for a bichromatic point set ${R}\cup{ B}$ (where the positive and negative weighted points are treated as red and blue points, respectively) as the outer rectangle, and we can choose a rectangle ${\cal R}_2$ inside ${\cal R}_1$ to obtain an annulus ${\cal A}_1$ with minimum penalty among all possible annuli with ${\cal R}_1$ as outer rectangle. However, there may exist a rectangle ${\cal R}_{out}$ (evidently, the penalty of ${\cal R}_{out}$ is larger than that of ${\cal R}_1$), and most of the blue points with huge penalties lie near the center of the rectangle ${\cal R}_{out}$. In that case, we choose a rectangle ${\cal R}_{in}$ circumscribing all such blue points so that ${\cal R}_{in}$ lies inside ${\cal R}_{out}$ and we may obtain an annulus ${\cal A}'$ with a penalty less than ${\cal A}_1$. 
We can reduce the $GRBAC$ problem to Maximum Weight Annulus problem as follows, although it helps to solve only the \colb{non-uniform} \colb{non-concentric} rectangular annulus in $GRBRAC$ problem.\\
\noindent {\bf Reduction:}  The penalty $\lambda$ of an annulus $\cal A$ is given by  

$$\lambda = 
    \displaystyle\sum_{q \in {R}^{out}} {\cal P}(q)
     + \displaystyle\sum_{p \in {B}^{in}} {\cal P}(p)={\cal P}(R)-\displaystyle\sum_{q \in {R}^{in}} {\cal P}(q)
     + \displaystyle\sum_{p \in {B}^{in}} {\cal P}(p)$$

\noindent where $R^{in}$ (resp. $B^{in}$) is the set of red (resp. blue) points lying inside the annulus $\cal A$. If we assign ${\cal P}(q)$ as the weight $w(q)$ of a red point $q$ and (-${\cal P}(p))$ as the weight $w(p)$ of a blue point $p$, then the above equation reduces to $\lambda = {\cal P}(R)-(\displaystyle\sum_{q \in {R}^{in}} {w}(q)
     + \displaystyle\sum_{p \in {B}^{in}} {w}(p))$. Since ${\cal P}(R)$ is constant, this equation implies that to minimize $\lambda$, we need to find an annulus that maximizes the total weight of the points inside it. In other words, $GRBAC$ problem reduces to a Maximum-Weight Annulus problem.

 However, to the best of our knowledge, there exists no prior work on the Maximum-Weight Annulus problem. 
 %Although there exists related problem of computing the maximum weighted rectangle for a weighted point set, it cannot be applied to solve our problems. 
For a given set of weighted points, we can compute a rectangle with maximum weight~\cite{ipl/BarbayCNP14} and a circle with maximum weight~\cite{iccsa/BeregDZR15}. However, these studies are based on rectangle and circle instead of the annulus, and the result of these problems cannot be applied to obtain the solution of our $GRBRAC$ or $GRBCAC$ problem, except for only one version (generalized version of the \colb{non-uniform} \colb{non-concentric} rectangular annulus) of $GRBRAC$ problem which is explained in the following paragraph.

In {\bf generalized version} of the {\bf \colb{non-uniform} \colb{non-concentric} rectangular annulus}, if we fix the outer rectangle $R_{out}$ (which passes through four red points), then we can compute an inner rectangle $R_{in}$ with maximum weights~\cite{iccsa/BeregDZR15} for the points lying inside $R_{out}$ in $O((m+n)^2)$ time so that the penalty of the annulus defined by $R_{out}$ and $R_{in}$ is minimized. Now we can iterate this procedure over all possible $R_{out}$ to obtain the optimal annulus (minimum penalty) in $O(n^4(m+n)^2)$ time, since the outer rectangle $R_{out}$ is defined by at most $4$ blue points. However, this procedure does not work for the other variations of the $GRBRAC$ problem, where the two rectangles $R_{out}$ and $R_{in}$ of the annulus $\cal A$ are concentric.

In the generalized version of {\bf circular annulus} ($GRBCAC$ problem), if we fix the outer circle $C_{out}$, then we can compute an inner circle $C_{in}$ with maximum weight~\cite{iccsa/BeregDZR15} for the points inside $C_{out}$, so that the penalty of the region bounded by $C_{out}$ and $C_{in}$ is minimized. However, it is not necessary that $C_{in}$ is concentric with $C_{out}$, and in this case, $C_{out}$ and $C_{in}$ do not form a circular annulus (as per the definition of circular annulus). Thus, this procedure also cannot generate the optimal solution for the generalized version of {\bf circular annulus} ($GRBCAC$ problem). 

So, in our study, we use the aforesaid reduction only for the computation of the \colb{non-uniform} \colb{non-concentric} rectangular annulus in the $GRBRAC$ problem, where there are no restrictions for the positions of the centers of ${\cal R}_{out}$ and ${\cal R}_{in}$.

%Now, one can say that if we first compute the inner circle ${C_{in}}$ with maximum weight~\cite{iccsa/BeregDZR15}, then we can construct the outer circle $C_{out}$ keeping it concentric with $C_{in}$ so that the penalty of the annulus made by $C_{out}$ and $C_{in}$ is minimized. However, in that case, $C_{out}$ always passes through a single~\cite{BergCKO08} red point, whereas there may exist a circular annulus of the minimum penalty whose $C_{out}$ is defined by two or three points. Thus, this procedure does not necessarily generate the optimal solution for $GRBCAC$. 

\section{Rectangular Annulus}
\label{joco0}
\subsection{One-Dimensional Red-Blue Annulus Cover Problem}
\label{joco1}
We consider the rectangular annulus cover problem in one dimension where the red-blue point-set lies on a given straight line $L$ (assuming $L$ is horizontal). We also assume that the centers of the two rectangles of the annulus ${\cal A}$ lie on the line $L$. It is to be noted that, in one dimension, the rectangular annulus $\cal A$ defines two non-overlapping intervals, say the left interval ${\cal I}_L = [L_o,L_i]$ and the right interval ${\cal I}_R= [R_i,R_o]$, where $L_o$ (resp. $R_o$) and $L_i$ (resp. $R_i$) are the left (resp. right) outer and the left (resp. right) inner endpoints (see Figure~\ref{fig:interval}). We use $p\in {\cal I}$ to denote that the point $p$ is covered by the interval ${\cal I}$. Length of an interval ${\cal I}$ is denoted by $|{\cal I}|$. The distance between any two points $a$ and $b$ is $|ab|$. 
  
\begin{figure}[ht]
 \captionsetup{belowskip=-5pt}
    \centering
    \includegraphics[scale=0.9]{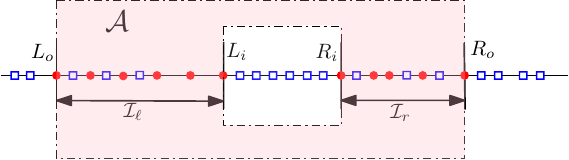}
    \caption{Rectangular annulus $\cal A$ defines two non-overlapping intervals ${\cal I}_L$ and ${\cal I}_R$.}
    \label{fig:interval}
\end{figure}
In the rest of this section, a rectangular annulus ${\cal A}$ indicates a pair of non-overlapping intervals ${\cal I}_L$ and ${\cal I}_R$. We say a pair of intervals to be {\bf uniform} if they are equal in length; otherwise, they are {\bf non-uniform}.
The ${\cal I}_L$ and ${\cal I}_R$ are said to be optimal if their union (${\cal I}_L \cup {\cal I}_R$) covers all the red points and a minimum number of blue points in the $RBRAC$ problem or have a minimum penalty in $GRBRAC$ problem, among all possible pairs of intervals. 
We discuss four different variations of this problem based on the choice of penalty for the red and blue points and the length of the intervals.
As a pre-processing task, we sort and store all the blue points ($p_1, p_2, \ldots, p_{m} \in B$) as well as all the red points ($q_1, q_2, \ldots, q_{n} \in R$) in two separate arrays. 
\begin{obsv}
 \label{each}
If the intervals are {\bf non-uniform}, then each endpoint of the optimal pair (${\cal I}_L$, ${\cal I}_R$) must be on some red point in ${R}$, provided each interval (${\cal I}_L$ and ${\cal I}_R$) covers at least two red points. As a case of degeneracy, only one red point, say $q_i$, may exist in an interval, say ${\cal I}_R$, of very small width $\epsilon$ ($>0$) whose one endpoint coincides with $q_i$.
\end{obsv}

\begin{obsv}
\label{blue_bound}
In case of {\bf uniform} intervals ${\cal I}_L$ and ${\cal I}_R$, the endpoints of one of them must be any two red points in $R$.
\end{obsv}
{\bf Justification of Observation 2:}
Suppose for the sake of contradiction, only one endpoint, say $L_o$ of the interval ${\cal I}_L$ of the optimal pair (${\cal I}_L$ and ${\cal I}_R$) is a red point in $R$. If the endpoints of ${\cal I}_R$ do not coincide with any point in $R \cup B$, then we shift ${\cal I}_R$ so that its one endpoint, say ${R}_i$ lies on a blue point without changing the penalty of ${\cal I}_R$. 
Now anchoring the ${\cal I}_L$ and ${\cal I}_R$ at their endpoint $L_o$ and $R_i$, respectively, we decrease both the ${\cal I}_L$ and ${\cal I}_R$ by equal length until a red point coincides with either $L_i$ or $R_o$. Shortening of ${\cal I}_L$ and ${\cal I}_R$ may cause a reduction in their penalty if a blue point is eliminated from ${\cal I}_L$ or ${\cal I}_R$.  If $L_i$ coincides with a red point, then we are done; otherwise, anchoring ${\cal I}_R$ at $R_o$, we further decrease ${\cal I}_L$ and ${\cal I}_R$ until a red point coincides with the endpoint $R_i$ or $L_i$, which contradicts our assumption, and hence the statement is proved. 

\subsubsection{Non-uniform Annulus (RBRACN-1D)}
\label{non_uniform}
This problem computes two non-overlapping non-uniform intervals ${\cal I}_L$ and ${\cal I}_R$ whose union covers all the red points in $R$ ($|R|$=$n$) and a minimum number of blue points in $B$ ($|B|$=$m$). 

\begin{obsv}
\label{conred}    
The two endpoints $L_i$ and $R_i$ must be any two consecutive red points $q_i$ and $q_{i+1}$, $2\leq i \leq {n-2}$.
\end{obsv}

{\bf Algorithm:}\\[0.1in]
We take the left endpoint $L_o$ of ${\cal I}_L$ at $q_1$ and the right endpoint $R_o$ of ${\cal I}_R$ at $q_{n}$ so that none of the red points lie to the left of ${\cal I}_L$ and right of ${\cal I}_R$ (see Observation~\ref{each}). Now we need to choose the other two endpoints $L_i$ and $R_i$ of ${\cal I}_L$ and ${\cal I}_R$ at some red points, respectively, so that ${\cal I}_L\cup{\cal I}_R$ covers all the red points in $R$ and minimum number of blue points in $B$. In a linear scan, we compute the number of blue points lying inside each pair of consecutive red points. We select such a pair ($q_i,q_{i+1}$) that contains a maximum number of blue points and choose these red points $q_i$ and $q_{i+1}$ as the right endpoint $L_i$ of ${\cal I}_L$ and left endpoint $R_i$ of ${\cal I}_R$, respectively (see Observation~\ref{conred}). We report these two intervals ${\cal I}_L$ and ${\cal I}_R$ as a solution. \\

{\bf Degeneracy Case:} Only one red point lies inside ${\cal I}_L$ or $ {\cal I}_R$. \\

We handle the degeneracy case by counting the number of blue points lying inside the pairs ($q_1$, $q_2$) and ($q_{n-1}$, $q_n$). If one of such pairs, say ($q_1$, $q_2$), contains the maximum number of blue points among all possible consecutive pairs $(q_i,q_{i+1})$, for $1 \leq i \leq (n-1)$, then we report ($q_1$, $q_1+\epsilon$) and ($q_2$, $q_n$) as the intervals ${\cal I}_L$ and ${\cal I}_R$, respectively.
If the other pair $(q_{n-1},q_n)$ contains the maximum number of blue points, then we report ($q_1$, $q_{n-1}$) and ($q_n$, $q_n+\epsilon'$) as the two intervals. Note that we choose $\epsilon$ (resp. $\epsilon'$) such that the pair $(q_1,q_1+\epsilon)$ (resp. $(q_n,q_n+\epsilon')$) does not contain any blue points $\in B$.\\

\textbf{Proof of correctness:}  The correctness proof of our algorithm is based on the following
two claims.
\begin{clam}
    Our solution covers all the red points $\in R$.
\end{clam}
\begin{proof}
    The left endpoint of ${\cal I}_L$ and the right endpoint of ${\cal I}_R$ are the first and last red points $q_1$ and $q_n$, respectively. The right endpoint of ${\cal I}_L$ and the left endpoint of ${\cal I}_R$ are the two consecutive red points, say $q_i$ and $q_{i+1}$. This proves the result. 
\end{proof}

\newpage
\begin{clam}
    Our solution covers the minimum number of blue points $\in B$.
\end{clam}
\begin{proof}
 It discards all the blue points that lie before (resp. after) the leftmost (resp. rightmost) red point. It also discards the maximum number of blue points lying between the two consecutive red points. 
\end{proof}
 Hence, our solution is optimal. Since each point is read at most once and our algorithm needs only linear space to store all the points, we obtain the following result.

\begin{theo}
The RBRACN-1D problem can be solved optimally in $O(m+n)$ time and $O(m+n)$ space, after $O((m+n)\log(m+n))$ preprocessing time.
\end{theo}
 
\subsubsection{Uniform Annulus (RBRACU-1D)}
\label{uniform}
We compute two uniform intervals ${\cal I}_L$ and ${\cal I}_R$ that cover all red points and a minimum number of blue points. The Observation~\ref{blue_bound} holds for this problem. We consider the following two cases depending on the color of the endpoints of each interval.\\[0.1in]
{\bf Algorithm:}\\[0.1in]
\noindent {\bf Case (i): The endpoints of ${\cal I}_L$ are red.}    
We choose a red point $q_i$ such that $|q_1q_i|>|q_{i+1}q_n|$ and  $|q_1q_{i-1}|<|q_iq_n|$. Take the interval $[q_1,q_i]$ as ${\cal I}_L$ and compute
its corresponding interval ${\cal I}_R$ so that $|{\cal I}_R|=|{\cal I}_L|$, ${\cal I}_L\cap {\cal I}_R=\phi$, ${\cal I}_R$ covers the red points that are not covered by ${\cal I}_L$, and ${\cal I}_R$ covers the minimum number of blue points. We determine this ${\cal I}_R$ in the following way.\\
We sequentially search for a blue point, say $p_j$, that lie between $q_i$ and $q_{i+1}$, so that $|p_jq_n|$ is as large as possible and $|p_jq_n|\leq |{\cal I}_L|$. We take ${\cal I}_R$ with its right endpoint on $q_n$ and compute its left endpoint (which lies on or to the left of $p_j$) so that $|{\cal I}_R|=|{\cal I}_L|$. We count the number of blue points covered by ${\cal I}_L\cup {\cal I}_R$. Next, keeping the length of ${\cal I}_R$ constant, we shift it rightward so that one of its endpoints (left or right) coincides with a blue point, which occurs immediately next to the right of the corresponding previous endpoint (left or right), depending on whichever occurs earlier.
 Note that if one endpoint of such ${\cal I}_R$ is a blue point, then we shift ${\cal I}_R$ by a very small distance $\epsilon >0$ toward left or right to discard that blue point and reduce one of the blue points covered by ${\cal I}_R$. 
We update the number of blue points covered by the two intervals. The above process continues until the left endpoint of ${\cal I}_R$ crosses the red point $q_{i+1}$. For this ${\cal I}_L=[q_1,q_i]$, we choose that position of ${\cal I}_R$ where the number of blue points covered by the ${\cal I}_R$ is minimized.
       
Now, in the next iteration, we increase the length of ${\cal I}_L$ by keeping its right endpoint on the next red point $q_{i+1}$ and repeat the above steps to choose its corresponding ${\cal I}_R$ while minimizing the number of blue points covered. This process is repeated until the right endpoint of ${\cal I}_L$ coincides with $q_{n-1}$. 

\begin{lema}
\label{linear_time}
The overall running time (amortized) to compute the ${\cal I}_R$ for all such ${\cal I}_L$ is $O(m+n)$.
\end{lema}
\begin{proof}
At each step, either the left endpoint or the right endpoint of ${\cal I}_R$ moves rightward through the blue or red points in $R \cup B$, and no such point is reprocessed twice by the same endpoint of ${\cal I}_R$, and this proves the result. 
\end{proof}
{\bf Case (ii): The endpoints of ${\cal I}_R$ are red.}  
We can deal with this case similarly to Case (i). In this case, we take the right interval ${\cal I}_R=[q_i,q_n]$ with the red point $q_i$ satisfying $|q_1q_{i-1}|<|q_iq_n|$ and  $|q_1q_i|>|q_{i+1}q_n|$, and  compute its corresponding left interval ${\cal I}_L$ which minimizes the number of blue points covered. Then we increase the length of ${\cal I}_R$ toward left only by placing its left endpoint on a red point $q_{i-1}$ that lies before $q_i$  at each iteration and compute its corresponding ${\cal I}_L$. Such increase of ${\cal I}_R$ and computation of the corresponding ${\cal I}_L$, is repeated until ${\cal I}_R$ covers $\{q_n,q_{n-1},\ldots,q_2\}$. 

Finally, among all the pairs of the intervals generated in the above two cases, we report the one that covers the minimum number of blue points.\\

\noindent{\bf Proof of correctness:}  The correctness proof of our algorithm is based on the following claim.
\begin{clam}
\label{equal_proof}
    Our algorithm generates all possible pairs of uniform intervals that cover all the red points in $R$.
\end{clam}
\begin{proof}
    For the sake of contradiction, we assume that the optimal pair of disks (${\cal I}_L^{opt}$, ${\cal I}_R^{opt}$) is not generated by our algorithm. By the Observation~\ref{blue_bound}, one of these intervals must have two red points at its endpoints, and without loss of generality, we take that interval as ${\cal I}_L^{opt}$. Since our algorithm chooses all pairs of red points as one of the intervals ${\cal I}_L$ or ${\cal I}_{R}$, this ${\cal I}_L^{opt}$ must be chosen in an iteration. If none of the endpoints of ${\cal I}_R^{opt}$ is a point in $R\cup B$, then we can shift ${\cal I}_R^{opt}$ until one of its endpoints coincides with a red point or blue point, which our algorithm must have generated, and hence it contradicts the assumption. Note that the position of ${\cal I}_R$ in our solution may differ from that of ${\cal I}_R^{opt}$; however, they must cover exactly the same set of red and blue points. 
\end{proof}
Our algorithm generates all possible uniform annuli covering all the red points and the minimum number of blue points following the Observation~\ref{blue_bound}. This algorithm reports the pair covering the minimum number of blue points among all such $feasible$ solutions; thus, the solution is optimal.

The Lemma~\ref{linear_time} leads to the following theorem. 

 \begin{theo}
The  RBRACU-1D problem can be solved optimally in $O(m+n)$ time and $O(m+n)$ space, with $O((m+n)\log(m+n))$ preprocessing time.
 \end{theo}    
\subsection{One-Dimensional Generalized Red-Blue Annulus Cover Problem}
\label{joco2}
\subsubsection{Non-uniform Annulus (GRBRACN-1D)}
\label{non_uniform_penalty}
The two intervals ${\cal I}_L$ and ${\cal I}_R$ are non-uniform, and hence the Observation~\ref{each} also holds for this problem. In this case, each point in $R \cup B$ is associated with a positive penalty.\\[0.1in]
{\bf Algorithm:}\\[0.1in]
We sequentially process all the red points rightwards, starting from $q_1 \in R$. Consider a red point $q_i$. Let ${\cal U}_1$ (resp. ${\cal U}_2$) be the set of red points $\in R$ (resp. blue points $\in B$) that lie on or to the left of $x(q_i)$.
We compute a non-uniform, non-overlapping optimal pair of intervals (${\cal I}_L$ and ${\cal I}_R$) up to point $q_i$ such that the following function is minimized.

$$\lambda = {\cal P}({\cal I}_L \cup {\cal I}_R) =  \sum_{p \in {B}^{in}} {\cal P}(p)+\sum_{q \in {R}^{out}} {\cal P}(q)$$ \\where ${B}^{in} \subseteq {\cal U}_2$ is the set of blue points covered by ${\cal I}_L \cup {\cal I}_R$ and  ${R}^{out} \subseteq {\cal U}_1$ is the set of red points not covered by ${\cal I}_L \cup {\cal I}_R$. \\Note that $\sum_{q \in {R}^{out}} {\cal P}(q) =  {\cal P}(R) -\sum_{q'\in {R'}}{{\cal P}(q')}$, where ${\cal P}(R)$ is the sum of penalties of the red points $\in {\cal U}_1$ and $R'$ is the set of red points covered by ${\cal I}_L \cup {\cal I}_R$.
Now, we compute this function up to $q_n$.
We use an incremental approach to process each red point in ${R}$ sequentially, increasing the order of their $x$-coordinate.  We maintain the four intervals ${\cal I}_1$, ${\cal I}_L$, ${\cal I}_R$ and  ${\cal I}_{q_i}$ up to the red point~$q_i$ where the endpoints of each such interval are on some red points (see Observation~\ref{each}). While processing the $(i+1)^{th}$ red point $q_{i+1}$, we update these intervals if necessary. 
The significance of these intervals is as follows.

\begin{itemize}
\setlength\itemsep{0.5em}
    \item[(i)] A single interval, say ${\cal I}_1=[u,v]$, of minimum penalty among all possible intervals up to the point $q_i$, where $u, v \in { R}$ are the two endpoints of ${\cal I}_1$.
    \item[(ii)] A pair of intervals ${\cal I}_L= [a, b]$ and ${\cal I}_R= [c, d]$ so that the penalty of their union is minimum among all the pair of intervals up to the point $q_i$, i.e., an optimal pair (${\cal I}_L,~{\cal I}_R$) up to $q_i$. 
    \item[(iii)] An interval ${\cal I}_{q_i}$ having minimum penalty with its right endpoint constrained to coincide with $q_i$. 
\end{itemize} 

 The algorithm executes the following while processing the next red point $q_{i+1}$.\\
 We first determine ${\cal I}_{q_{i+1}}$. The optimal pair of intervals ${\cal I}_L$ and ${\cal I}_R$ up to $q_i$ either remains optimal or needs to be updated.  We compute the penalties of the following four pairs of intervals and return the optimal pair with a minimum penalty up to $q_{i+1}$.
 
 ( ${\cal I}_1$, ${\cal I}_{q_{i+1}}$), (${\cal I}_L$, ${\cal I}_{q_{i+1}}$), (${\cal I}_R$, ${\cal I}_{q_{i+1}}$), and (${\cal I}_L$, ${\cal I}=[c,q_{i+1}]$).
 
We must update the pair of intervals (${\cal I}_L$, ${\cal I}_R$) with the reported pair. We also update the single interval ${\cal I}_1$ to be used in the next iteration. Note that the width of one of the intervals can be $\epsilon$ (if the interval contains a single red point with a very large penalty), which corresponds to the degeneracy case.\\
  
 {\bf Proof of correctness:} The correctness proof of our algorithm is based on the following facts and observations.
 \begin{fact}
      If the optimal pair (${\cal I}_L$, ${\cal I}_R$) up to point $q_i$, needs to be updated while processing the point $q_{i+1}$, then the right endpoint of ${\cal I}_R$ (after update) must be the red point $q_{i+1}$.
 \end{fact}
%We compute ${\cal I}_{q_{i+1}}$. 
 \begin{obsv}
 \label{grb_proof_1}
   If the left endpoint of ${\cal I}_{q_{i+1}}$ does not overlap with ${\cal I}_1$, ${\cal I}_L$ and ${\cal I}_R$ in the previous iteration (i.e., up to $q_i$), then one of these intervals becomes the left interval ${\cal I}_L$ in the current iteration (i.e., up to $q_{i+1}$), and the ${\cal I}_{q_{i+1}}$ becomes the right interval ${\cal I}_R$ up to point $q_{i+1}$.  
 \end{obsv}
 The Observation~\ref{grb_proof_1} implies that we must compare  ${\cal I}_{q_{i+1}}$ with each one of the ${\cal I}_1$, ${\cal I}_L$ and ${\cal I}_R$.

\begin{obsv}
\label{grb_proof_2}
    If the left endpoint of ${\cal I}_{q_{i+1}}$ overlaps, with  ${\cal I}_R=[c,d]$, then the left endpoint of ${\cal I}_{q_{i+1}}$ will be $c$. 
\end{obsv}
 
  The Observation~\ref{grb_proof_2} says that, we must compare ${\cal I}=[c,q_{i+1}]$ with ${\cal I}_L$.

  \begin{obsv}
  \label{iq1}
       The interval ${\cal I}_{q_{i+1}}$ cannot overlap with ${\cal I}_L$ otherwise ${\cal I}_L$ and ${\cal I}_R$ would not have been optimal pair up to point $q_i$.
  \end{obsv}
  The Observation~\ref{iq1} says that we do not need to consider any point covered by ${\cal I}_L$ in the previous iteration (up to the point $q_i$) to compute the ${\cal I}_{q_{i+1}}$.\\
  Thus, our algorithm produces the correct result up to $q_{i+1}$, and hence, the optimal solution after processing $q_n$.
  
  To update the four intervals at each point $q_i\in R$, we need to consider the penalties of red point $q_i$ and the blue points lying between $q_{i-1}$ and $q_i$. Hence, the update operations require $O(k_i+1)$ time, where $k_i$ is the number of blue points lying between $q_{i-1}$ and $q_i$. Thus we obtain the following result.
  
 \begin{theo}
 We can compute the non-uniform annulus of minimum penalty in the GRBRACN-1D problem in $O(m+n)$ time using $O(m+n)$ space.
 \end{theo}
  
\subsubsection{Uniform Annulus (GRBRACU-1D)}
\label{uni}

We compute an optimal pair (i.e., of minimum penalty) of uniform intervals ${\cal I}_L=[L_o,L_i]$ and ${\cal I}_R=[R_i,R_o]$. 
In this case, $|{\cal I}_L|=|{\cal I}_R|$. \\

\noindent{\bf Algorithm:}\\[0.1in]
Without loss of generality, we assume that the colors of both the endpoints of ${\cal I}_L$ are red (see Observation~\ref{blue_bound}). We do the following tasks for each pair of red points ($q_i,q_j$) in $R$.\\ 
We consider an interval ${\cal I}_L$ with its two endpoints at ($q_i,q_j$).
For this ${\cal I}_L$, we find an interval ${\cal I}_R$ lying to the right of ${\cal I}_L$  so that penalty of ${\cal I}_R$ is minimized. For this, first, we consider an interval ${\cal I}_L$ and then take another interval ${\cal I}$ of length $|{\cal I}_L|$ whose left endpoint coincides with a blue or red point that is not covered by ${\cal I}_L$. Then we shift this interval ${\cal I}$  rightward sequentially, either with its left or right endpoint coinciding with a blue or red point that lies immediately next to the left endpoint or right endpoint of $\cal I$, depending on whichever occurs first. If one endpoint of such ${\cal I}$ is a blue point, then we shift it by a very small distance $\epsilon >0$ toward the left or right to discard that blue point, thereby reducing the penalty. In this way, we compute the penalties of the intervals ${\cal I}$ with one of their endpoints being at each different red or blue point (which are not covered by ${\cal I}_L$). Among all such ${\cal I}$, we choose the one with minimum penalty as ${\cal I}_R$ for the given ${\cal I}_L$.\\
Similarly, we can repeat the above tasks to search for a ${\cal I}_L$ by taking both the endpoints of ${\cal I}_R$ as red points. Finally, we choose the pair with minimum penalty.\\

\noindent{\bf Proof of correctness:}  The correctness proof of our algorithm is based on the following
claim.
\begin{clam}
\label{equal_proof_penalty}
    Our algorithm generates all possible pairs of uniform intervals that cover all the red points in $R$, and all these annuli satisfy the Observation~\ref{blue_bound}.
\end{clam}
\begin{proof}
The proof is similar to that of the Claim~\ref{equal_proof}.
\end{proof}

Our algorithm reports the pair with the minimum penalty, thus providing the optimal solution.
The above procedure needs $O(m+n)$ time. Since there are $O(n^2)$ distinct positions of ${\cal I}_L$, we obtain the following result
 \begin{theo}
 We can compute the uniform annulus with a minimum penalty in the GRBRACU-1D problem in $O(n^2(m+n))$ time with $O(m+n)$ space.
 \end{theo}

\subsection{Two Dimensional Red-Blue Annulus Cover Problem}
\label{joco3}
We compute an annulus $\cal A$ for a given set of bichromatic points lying on $\mathbb{R}^2$, which covers all the red points and the minimum number of blue points. 
We denote the left-most (resp. right-most) red point in ${R}$ by $q_\ell$ (resp. $q_r$), and the bottom-most (resp. top-most) red point in ${R}$ by $q_b$ (resp. $q_t$). An annulus $\cal A$ is said to be {\color{blue} $feasible$} if it covers all the red points in $R$. Among all $feasible$ annuli, the one that covers the minimum number of blue points is called {\color{blue} $minimum$-${\cal A}$} and is denoted by ${\cal A}_{\min}$.
Four points are sufficient to identify a rectangle uniquely. Its two opposite corner points can also define it.
The two rectangles ${\cal R}_{out}$ and ${\cal R}_{in}$ of the annulus ${\cal A}$ may or may not be concentric. Depending on the concentricity and widths of the annulus, we study non-uniform, non-concentric, non-uniform concentric, and uniform annulus cover problems. The properties of all such annuli are stated below.
\newline
The four widths of a {\color{blue} non-uniform non-concentric annulus} are different (see Section~\ref{prelim}), and  we obtain the following result.
\begin{lema}
\label{non_con}
Eight red points are sufficient to uniquely define the boundaries of a non-uniform non-concentric rectangular annulus that covers the minimum number of blue points.
\end{lema}
\begin{proof}
In a non-uniform non-concentric annulus $\cal A$, the four widths $w_\ell$, $w_r$, $w_t$ and $w_b$ are different. To define $w_\ell$ (resp. $w_r$), we need two red points on the left (resp. right) side of both ${\cal R}_{out}$ and ${\cal R}_{in}$, otherwise the width can be minimized further. Similarly, two red points on the top (resp. bottom) side of both ${\cal R}_{out}$ and ${\cal R}_{in}$ are required to define $w_t$ (resp. $w_b$). Therefore, a total of eight red points are sufficient to uniquely define a  non-uniform non-concentric annulus $\cal A$ covering the minimum number of blue points. 
\end{proof}
In a non-uniform concentric annulus, the horizontal width ($w_h$) and the vertical width ($w_v$) are not equal. So we obtain the following result.
\begin{lema}
\label{lm1}
Six points are sufficient to define the boundaries of a non-uniform concentric rectangular annulus uniquely, and among these six points, four must be red so that the annulus covers the minimum number of blue points.
\end{lema}
\begin{proof}
 Both the horizontal and the vertical width of the non-uniform concentric annulus ${\cal A}$ must be determined by two red points each; otherwise, its widths can be minimized further. To define the horizontal width $w_h$ of the $\cal A$, two red points must lie on the $left({\cal R}_{out})$ and $left({\cal R}_{in})$ or on the $right({\cal R}_{out})$ and $right({\cal R}_{in})$. Similarly, for the vertical width $w_v$ of the ${\cal A}$ to be defined, the two red points lie on the top or bottom boundaries of ${\cal R}_{out}$ and ${\cal R}_{in}$. If these two red points determine the left (resp. right) width, then a single point of any color lying on the right (resp. left) boundaries of ${\cal R}_{out}$ or ${\cal R}_{in}$, is needed to determine the corresponding right (resp. left) width. So, three points must lie on any three vertical sides of ${\cal R}_{out}$ and ${\cal R}_{in}$ together. Similarly, three points lie on any three of the horizontal sides of ${\cal R}_{out}$ and ${\cal R}_{in}$ together. Hence, a total of six points are sufficient to uniquely define a non-uniform concentric annulus $\cal A$, out of which four must be red. 
\end{proof}
All four widths of a {\color{blue} uniform annulus} are the same, and hence, we obtain the following result.
\begin{lema}
\label{lm2}
Five points are sufficient to define the boundaries of a uniform rectangular annulus uniquely, and among these five points, two must be red so that the annulus covers the minimum number of blue points.
\end{lema}
\begin{proof}
 The width, say $w$, of a uniform rectangular annulus must be determined by two red points lying on the same side of both ${\cal R}_{in}$ and ${\cal R}_{out}$; otherwise, its width can be minimized further. Without loss of generality, suppose these two red points lie on the left sides of ${\cal R}_{out}$ and ${\cal R}_{in}$. Also note that at least one point of any color must lie on the top side (resp. right side and bottom side) of ${\cal R}_{out}$ or ${\cal R}_{in}$ to define the $Top({\cal A})$ (resp. $Right({\cal A})$ and $Bottom(\cal A)$). This proves the result.     
\end{proof}
\begin{obsv}
    \label{rru_an}
    If any of the four-sided regions $(Top(\cal A)$, $Bot(\cal A)$, $Left(\cal A)$ and $Right(\cal A))$ of a uniform annulus ${\cal A}$ is defined by one or two blue points, then it is possible to shift that region(s) (without changing its width(s)) by a very small distance $\epsilon>0$ to obtain another uniform annulus ${\cal A'}$ so that $\cal A'$ covers exactly the same set of points that are covered by $\cal A$, except the blue point that defines any of the four-sided regions e.g., $Top(\cal A)$ and hence $\cal A'$ has less penalty than ${\cal A}$. Note that $Top(\cal A)$ and $Bot(\cal A)$ are to be shifted upward or downward, while $Left(\cal A)$ and $Right(\cal A)$ need to be shifted leftward or rightward.
\end{obsv}

Now we discuss the algorithms for each different type of red-blue rectangular annulus cover problem in the following subsections. To compute the optimal result, we need to generate all $feasible$ annuli for the point-set $R\cup B$.\\

{\bf Remarks: } While generating the annuli, if any one of its boundaries passes through a blue point, then we always shift the corresponding region of $\cal A$ by a very small distance $\epsilon>0$ to discard the blue point (s) on its boundaries and obtain an annulus $\cal A'$ of the same width with less number of blue points (see  Observation~\ref{rru_an}). Once such an annulus $\cal A$ is generated, we shift it before counting the number of blue points inside it, provided its boundary passes through the blue point(s), and we always assume that this shifting operation (by $\epsilon$) is executed in our algorithm, and do not mention it to avoid saying this repetitively.
\subsubsection{Non-uniform Non-concentric Annulus (RBRACNNC-2D)}
\label{non_uniform1}
In this problem, the annulus ${\cal A}$ is non-concentric, and its width is non-uniform, i.e., all the widths are different. Since ${\cal A}$ covers all the red points, we observe the following.
\begin{obsv}
\label{out}
    The left, right, top, and bottom sides of ${\cal R}_{out}$ are defined by 
    leftmost, rightmost, topmost, and bottom-most red points i.e., $q_\ell$, $q_r$, $q_t$ and $q_b$, respectively.
\end{obsv}
As a preprocessing task, we store all the red points in the two separate sorted arrays (one with respect to $x$-coordinate and the other with respect to $y$-coordinate). Similarly, we also store all the blue points sorted order in two arrays.\\[0.1in]
\noindent{\bf Algorithm:}\\[0.1in]
We consider the set of red points $\{q_i|~q_i\in {R} \setminus \{q_\ell,q_r,q_t,q_b\}\}$ to generate ${\cal R}_{in}$ of all $feasible$ annuli (as per the Observation~\ref{out}). 
For this, we take such a red point $q_i$ to construct all $feasible$ annuli with the left side of ${\cal R}_{in}$ being defined by that $q_i$. For such a $q_i$, we consider all the red points $q_j\in { R}\setminus \{q_\ell,q_r,q_t,q_b,q_i\}$ 
so that $q_j$ defines the right side of all ${\cal R}_{in}$ which are generated as follows.
We first construct an empty rectangle ${\cal R}_{in}$ being defined by two opposite corners $q_i$ and $q_j$. Then, we process the remaining red points sequentially in the increasing order of their $x$-coordinate to update the ${\cal R}_{in}$ and take each such red point as current $q_j$. We update the ${\cal R}_{in}$ so that its left and right sides always pass through $q_i$ and current $q_j$, respectively, and ${\cal R}_{in}$ does not contain any other red point inside it. We repeat the above procedure for all such $q_i$ and report the annulus with minimum number of blue points inside it as the ${\cal A}_{\min}$. 
\begin{lema}
\label{const_time}
  It needs a constant amount of time to update the ${\cal R}_{in}$.
\end{lema}
\begin{proof}
    While updating ${\cal R}_{in}$ with current $q_j$, we need to check only the four boundary points of ${\cal R}_{in}$ constructed with the red point $q_{j-1}$ in the previous iteration.
    If $q_j$ lies above the top side of ${\cal R}_{in}$ (see Figure~\ref{fig:const_fig}(c)) or below the bottom side of ${\cal R}_{in}$ (see Figure~\ref{fig:const_fig}(d)), then no update of ${\cal R}_{in}$ is possible since a red point appears inside ${\cal R}_{in}$.
    In other cases, we update ${\cal R}_{in}$ as shown in Figure~\ref{fig:const_fig}(a) and Figure~\ref{fig:const_fig}(b). Thus, it needs a constant amount of time to update ${\cal R}_{in}$. 
\end{proof}
\begin{figure}[ht]
 \captionsetup{belowskip=-5pt}
 \vspace{-0.1in}
    \centering
    \includegraphics[scale=0.85]{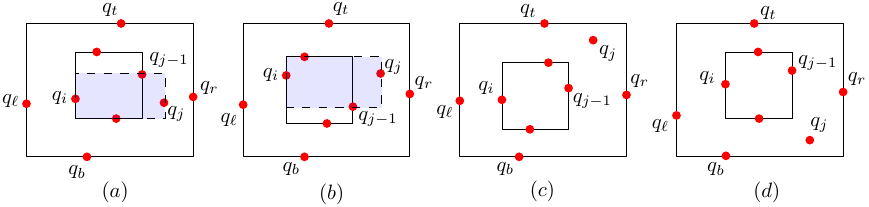}
    \caption{${\cal R}_{in}$ can be updated only in (a) and (b) as shown by the shaded rectangle.}
    \label{fig:const_fig}
\end{figure}

{\bf Proof of correctness:} The correctness proof of our algorithm is based on the following claim.
\begin{clam}
    Our algorithm generates all the $feasible$ annuli.
\end{clam}
\begin{proof}
    We generate the annulus $\cal A$ using the Observation~\ref{out} so that no red point lies outside ${\cal R}_{out}$. Also, we follow the Lemma~\ref{non_con} to compute all possible ${\cal R}_{in}$ in our algorithm so that no red points lie inside the boundary of ${\cal R}_{in}$. Hence, only the $feasible$ annuli are generated in our algorithm. For each red point $q_i$ lying inside the ${\cal R}_{out}$, we generate all possible ${\cal R}_{in}$ and thus the result is proved.  
\end{proof}
 Finally, our algorithm chooses the annulus with minimum blue points covered, and thus it is optimal.

The above algorithm computes all possible rectangles ${\cal R}_{in}$ with $q_i$ on its left side (containing no red points inside it) and counts the respective number of blue points covered by the annulus $\cal A$ in $O(m+n)$ amortized time (using the Lemma~\ref{const_time}). Since there are $O(n)$ red points, we obtain the following result.

\begin{theo}
For a set of bichromatic points lying on $\mathbb{R}^2$, the $RBRACNNC-2D$ problem can be solved in $O(n(m+n))$ time using $O(m+n)$ space.
\end{theo}

\subsubsection{Non-uniform Concentric Annulus (RBRACNC-2D)}
\label{non_uni_con}
%\noindent
In this case, the horizontal width $w_h$ and the vertical width $w_v$ of the annulus ${\cal A}$ are different. Among the six defining points of $\cal A$ (see Lemma~\ref{lm1}), either four points lie on the ${\cal R}_{in}$ and the remaining two points lie on ${\cal R}_{out}$ (see Figure~\ref{fig3}(a)) or vice versa (see Figure~\ref{fig3}(c)), or three points lie on both ${\cal R}_{in}$ and ${\cal R}_{out}$ (see Figure~\ref{fig3}(b)).
\begin{figure}%[ht]
    \centering
\includegraphics[width=\textwidth]{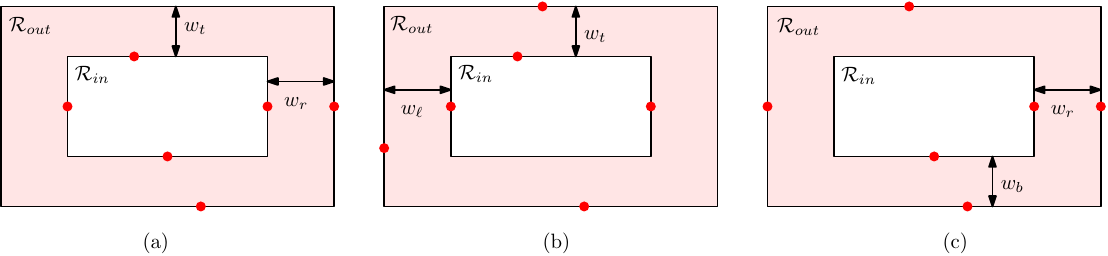}
    \caption{Instance of rectangular annulus defined by six points.}
    \label{fig3}
\end{figure}

\begin{fact}
\label{fact2}
Since the annulus $\cal A$ covers all the red points, the left (resp. right) side of the ${\cal R}_{out}$ cannot lie to the right (resp. left) of $q_\ell$ (resp. $q_r$), and the top (resp. bottom) side of the ${\cal R}_{out}$ cannot lie below (resp. above) the $q_t$ (resp.~$q_b$). 
\end{fact}
Fact~\ref{fact2} states that one of the following four cases must be true for the annulus~$\cal A$.

\begin{description}
\setlength\itemsep{0.5em}
    \item[{\bf Case (i)}] The $left({\cal R}_{out})$ and $bot({\cal R}_{out})$ pass through $q_\ell$ and $q_b$, respectively.
    \item[{\bf Case (ii)}] The $left({\cal R}_{out})$ and $top({\cal R}_{out})$ pass through $q_\ell$ and $q_t$, respectively. 
    \item[{\bf Case (iii)}] The $right({\cal R}_{out})$ and $bot({\cal R}_{out})$ pass through $q_r$ and $q_b$, respectively. 
    \item[{\bf Case (iv)}] The $right({\cal R}_{out})$ and $top({\cal R}_{out})$ pass through $q_r$ and $q_t$, respectively. 
   \end{description}
\noindent{\bf Algorithm:}\\[0.1in]
First, we generate all possible annuli for Case (i) and Case (ii) together as follows:
\\
We use two vertical sweep lines $L_1$ and $L_2$, where $L_2$ lies to the right of $L_1$ and these two lines generate the left and right sides of ${\cal R}_{in}$, respectively. These lines sweep in the rightward direction. The red points (resp. both red and blue points) lying to the right of $x(q_\ell)$ and above $y(q_b)$ are the event points for the sweep lines $L_1$ (resp. $L_2$). We also take two vertical lines $V_1$ and $V_2$ passing through the left and right boundaries of the ${\cal R}_{out}$, respectively.
\begin{figure}
    \centering
    \includegraphics[scale=0.7]{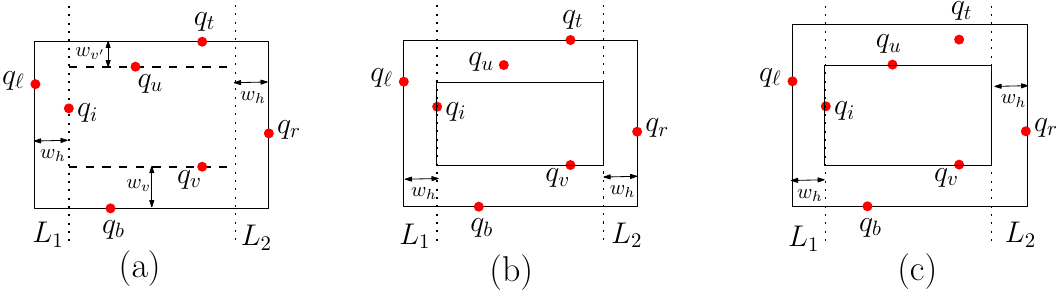}
    \caption{Instance of Case (i) in $RBRACNC-2D$ problem.}
    \label{rbacn}
\end{figure}
\begin{fact}
\label{fact3}
Among the four red points (see Lemma~\ref{lm1}), two red points lie on any two adjacent sides of the ${\cal R}_{out}$, and the same adjacent sides of ${\cal R}_{in}$ also pass through the other two red points so that both the widths $w_h$ and $w_v$ are defined.
\end{fact}
Using the Fact~\ref{fact3}, we first generate all possible annuli with the horizontal width $w_h$ defined by $q_\ell$ and $q_i$, i.e., $w_h=|x(q_i)-x(q_\ell)|$, where the red point $q_i$ lies to the right of $q_\ell$. We take $L_1$ passing through $q_i$.  Next, we take the vertical sweep line $L_2$ at a distance $w_h$ to the left of $x(q_r)$ (see Figure~\ref{rbacn}). Then, we choose two red points $q_u$ and $q_v$ of any color within the region bounded by the two lines $L_1$ and $L_2$ in such a way that $q_u$ and $q_v$ lie immediately above and below the horizontal line passing through $q_i$, respectively. Let $w_v=|y(q_v)-y(q_b)|$ and $w_{v'} = |y(q_t)-y(q_u)|$ (see Figure~\ref{rbacn}(a)). We construct rectangular annuli of vertical width $w_v$ or $w_{v'}$ depending on the following conditions. 
\begin{enumerate}
  \item[(i)] {\bf \boldmath $w_{v'}<w_v$}: In this case, we generate the annuli of vertical width $w_v$.
    First, we create an annulus ${\cal A}$ of width $w_v$ with ${\cal R}_{out}$ passing through $q_\ell$, $q_t$, $q_r$ and $q_b$. The $bot({\cal R}_{in})$ of $\cal A$ passes through $q_v$ and $top({\cal R}_{in})$ lies below $q_u$ (see Figure~\ref{rbacn}(b)). Maintaining the width of $Top(\cal A)$ as $w_v$, we shift the $top({\cal R}_{out})$ (resp. $top({\cal R}_{in})$) upwards sequentially through all the sorted blue points lying within $V_1$ and $V_2$ (resp. $L_1$ and $L_2$) to generate all possible annuli until the $top({\cal R}_{in})$ passes through $q_u$ (see Figure~\ref{rbacn}(c)). We count the number of blue points covered by all such annuli. Note that we generate all the annuli sequentially one after another.
    Once we generate an annulus, we obtain the number of blue points covered by it by updating that of the previous annulus just by adding or subtracting (depending on its color) the number of blue points that enter or leave the new annulus, 
    and hence it needs constant time. In this way, we generate all possible annuli of width $w_v$ with $left({\cal R}_{in})$ and $right({\cal R}_{in})$ being defined by  $L_1$ and $L_2$, respectively.\\
     \item[(ii)] {\bf \boldmath $w_{v'}>w_v$}:
      In this case, no $feasible$ annulus $\cal A$ of vertical width $w_v$ with its $left({\cal R}_{in})$ and $right({\cal R}_{in})$ being defined by $L_1$ and $L_2$, respectively, is possible. However, we can generate the annuli of vertical width $w_{v'}$, by extending and shifting the $Bot(\cal A)$ in an analogous manner as described above in the previous condition (where $w_{v'}<w_v$).
\end{enumerate}
Note that if there are no such red points $q_u$ and $q_v$ inside the region bounded by $L_1$ and $L_2$, then we take the top (resp. bottom) side of the ${\cal R}_{in}$ at a small distance $\epsilon>0$ below (resp. above) the top (resp. bottom) side of ${\cal R}_{out}$.

In the next iteration, the line $L_2$ sweeps rightward to its next event point, say $s$ (may be red or blue colored), and we shift the $Right({\cal A})$ rightwards keeping its width $w_h$. Then, we repeat the above procedure again to compute all the $feasible$ annuli of horizontal width $w_h$, which is realized by two red points lying on the $left({\cal R}_{in})$ and $left({\cal R}_{out})$. We execute the above algorithm by choosing all possible red points $q_i \in R \setminus \{q_\ell,q_r,q_t,q_b\}$ to generate the non-uniform concentric annuli with all possible horizontal widths (determined by two red points), i.e., all the annuli for Case (i) and Case (ii). 

In the same way, we generate the annuli of all possible horizontal width $w_h$ determined by two red points lying on the $right({\cal R}_{in})$ and $right({\cal R}_{out})$. In other words, all the $feasible$ annuli for Case (iii) and Case (iv) are generated. Finally, among all such annuli generated, we report the one with the minimum number of blue points as ${\cal A}_{\min}$.\\[0.1in]
{\bf Proof of correctness:} The correctness proof of our algorithm is based on the following claim.
\begin{clam}
    Our algorithm generates all the $feasible$ annuli and provides the optimal solution.
\end{clam}
\begin{proof}
Our algorithm chooses all possible pairs of red points ($q_i$, $q_j$), where $q_i,q_j \in R\setminus\{q_\ell,q_t,q_r,q_b\}$ as the left side and the right side of ${\cal R}_{in}$ in such a way that the rectangle ${\cal R}_{in}$ contains no other red points inside it, and hence only $feasible$ annulus are generated. We generate all possible annuli following the Fact~\ref{fact2} and Fact~\ref{fact3}.   
For the sake of contradiction, we assume that our algorithm does not generate the optimal solution. 
  The optimal solution must have four red points on its boundaries that satisfy the property mentioned in the Fact~\ref{fact3}, and these four red points of the optimal annulus must have also been chosen by our algorithm for the same purpose as stated in the Fact~\ref{fact3}. We determined all possible annuli for each of the four cases (followed by Fact~\ref{fact2}), and hence the optimal annulus must be reported by our algorithm. This contradicts the assumption, and hence the result is proved.
  
\end{proof}

\begin{theo}
For a set of bichromatic points in $\mathbb{R}^2$, the $RBRACNC-2D$ problem can be solved in $O(n(m+n)^2)$ time along with $O(m+n)$ space.
\end{theo}
\begin{proof}
 We have presented the algorithm and proved its correctness above. We determine its time complexity as follows. For a fixed position of $Left(\cal A)$ and $Right(\cal A)$ with horizontal width $w_h$, the computation of all possible annuli with different vertical widths needs $O(m+n)$ time. We shift $Right(\cal A)$ at most $(m+n)$ times rightwards. Hence the computation of all possible annuli with horizontal width determined by the two red points on $left({\cal R}_{in})$ and $left({\cal R}_{out})$ needs $O(m+n)^2$. Since there are $O(n)$ possible choices of $q_i$, it requires $O(n(m+n)^2)$ time to compute all possible annuli of distinct widths determined by $Left({\cal A})$ with two red points. Similarly, the annuli constructed in the other cases also need the same time, and hence the result is proved.
    
\end{proof}
\subsubsection{Uniform Annulus (%\texorpdfstring
{$RBRACU$-$2D$})}
\label{uni-2d}
Two red points that determine the width of the uniform annulus must lie on the same side (e.g., left side) of the ${\cal R}_{out}$ and ${\cal R}_{in}$ (see Lemma~\ref{lm2}) and the other three points define the other three regions of $\cal A$ (e.g., $Top(\cal A)$, $Right(\cal A)$, and $Bot(\cal A)$).\\[0.1in]
\noindent{\bf Algorithm:}\\[0.1in]
For each red point $q_i\in R\setminus\{q_\ell,q_t,q_r,q_b\}$, we do the following tasks:

We take four lines $V_1$, $V_2$, $L_1$ and $L_2$ similar to our algorithm for the problem $RBRACNC-2D$ described in the Section~\ref{non_uni_con}, and find out the two respective red points $q_u$ and $q_v$. We compute $w_v=|y(q_v)-y(q_b)|$ and $w_{v'}=|y(q_t)-y(q_u)|$. 
We compute all possible uniform annuli of width $w=|x(q_i)-x(q_\ell)|$, which is determined by two red points $q_\ell$ and $q_i$ lying on the left boundaries of the ${\cal R}_{out}$ and ${\cal R}_{in}$, as follows.
\begin{description}
    \item[(i)] {\bf\boldmath $w<w_{v'}$ or $w<w_v$}: It is not possible to construct any uniform annulus of width $w$ determined by the two red points $q_\ell$ and $q_i$. 
    \item[(ii)] {\bf\boldmath $w_{v'}<w$ and $w_v<w$}: We take two horizontal lines, one at a $w$ distance below the $x(q_t)$, and the other at a $w$ distance above the $x(q_b)$, which represent the $top({\cal R}_{in})$ and $bot({\cal R}_{in})$ respectively, and we obtain a uniform annulus of width $w$ and count the number of blue points inside it. To minimize the number of blue points covered by the annulus, we choose the positions of $Top({\cal A})$ and $Bot(\cal A)$, which are bounded by the pairs ($V_1$, $V_2$) and ($L_1$, $L_2$), as follows:\\ 
    Keeping the width of $Top(\cal A)$ same, i.e., $w$, we shift the $Top({\cal A})$ upwards sequentially through the blue points so that either $top({\cal R}_{out})$ or $top({\cal R}_{in})$ passes through a blue point lying within the pair ($V_1$, $V_2$) or ($L_1$, $L_2$) to generate the other annuli of same width until the $top({\cal R}_{in})$ passes through $q_u$. While shifting, we update the number of blue points inside $\cal A$. Among all such annuli, only the position of $Top(\cal A)$ differs, and we choose the one so that the number of blue points inside $\cal A$ is minimized. 

    Similarly, we do the same tasks to find out the proper position of $Bot(\cal A)$ so that the number of blue points inside $\cal A$ is further minimized. However, in that case, the $Bot(\cal A)$ is to be shifted downwards sequentially through the blue points keeping its width the same (i.e., $w$) until the $bot({\cal R}_{in})$ passes through $q_v$.

    Now, maintaining the same width, we shift the $Right(\cal A)$ rightwards sequentially until the $right({\cal R}_{in})$ passes through the $q_r$, and for each different positions of $Right(\cal A)$, we repeat the above tasks and compute the annulus with the minimum number of blue points whose width is determined by $Left(\cal A)$.
\end{description}

Now, we repeat the above tasks to compute all possible uniform annuli of width determined by two red points lying on the right boundaries (resp. top boundaries,  bottom boundaries) of the ${\cal R}_{out}$ and ${\cal R}_{in}$, and report the one with the minimum number of blue points covered by the annulus.\\[0.1in]
\noindent {\bf Proof of correctness:} The correctness proof of our algorithm is based on the following claim.
\begin{clam}
    Our algorithm generates all the $feasible$ annuli and provides the optimal solution.
\end{clam}
\begin{proof}
  For the sake of contradiction, we assume that our algorithm does not generate the optimal solution. 
  The optimal solution, say ${\cal A}_{opt}$ must have two red points on the same (left/right/top/bottom) side of both ${\cal R}_{out}$ and ${\cal R}_{in}$ to determine the width of ${\cal A}_{opt}$. We choose all the red points  $q_i\in (R\setminus \{q_\ell,q_t,q_r,q_b\})$ and then generate all $feasible$ uniform annuli with width determined by $q_i$ and each one of  $\{q_\ell,q_t,q_r,q_b\}$, following the Lemma~\ref{lm2}. Hence, the two red points of the ${\cal A}_{opt}$ must have been chosen by our algorithm in an iteration, and then it computes all possible annuli, including the annulus $\cal A$ that covers the minimum number of blue points. Thus, our algorithm also reports the optimal annulus, which contradicts the assumption, and hence, it proves the result. 
\end{proof}

\begin{theo}
For a set of bichromatic points lying on $\mathbb{R}^2$, the $RBRACU-2D$ problem can be solved optimally in $O(n(m + n)^2)$ time and $O(m+n)$ space.
\end{theo}
\begin{proof}
We described the algorithm along with its correctness proof above. We determine its time complexity as follows. For a fixed position of $Left(\cal A)$ and $Right(\cal A)$ with width $w$, the computation of all possible annuli with the same width needs $O(m+n)$ time. We shift $Right(\cal A)$ at most $(m+n)$ times rightwards. Hence the computation of all possible annuli with the width $w$ determined by the two red points on $left({\cal R}_{in})$ and $left({\cal R}_{out})$, needs $O(m+n)^2$. Since there are $O(n)$ possible choices of $q_i$, it requires $O(n(m+n)^2)$ time to compute all possible annuli of distinct widths determined by $Left({\cal A})$ with the two red points lying on the boundaries of $Left(\cal A)$. Similarly, the annuli constructed in the other cases also need the same time, and hence the result is proved. 
\end{proof}
\vspace{-0.2in}
\subsection{Generalized Red-Blue Annulus Cover Problem in  \texorpdfstring{$2$-$D$}~}
\label{joco4}
 The bichromatic points in ${\mathbb R}^2$ are associated with different penalties. We compute the annulus $\cal A$ of non-uniform width with the minimum penalty, and there are two versions (non-concentric and concentric) of this problem depending on the positions of ${\cal R}_{out}$ and ${\cal R}_{in}$. Similar to various types of $RBRAC$ problems in $2D$ studied in Sections~\ref{joco3}, we study the corresponding generalized, (i.e., penalty) version problems in the following three subsections.

\vspace{-0.1in}
\subsubsection{Non-uniform Non-Concentric Annulus (\texorpdfstring{$GRBRACNNC-2D$}))}
\label{grbacn-2d-nc}
All the widths of the annulus are different.
In {\bf Remarks} (see Section~\ref{remarks}), we have already stated how this problem can be solved in $O(n^4(m+n)^2)$ time by computing a rectangle with maximum weight~\cite{ipl/BarbayCNP14}. Thus, we obtain the following result.

\begin{theo}
The problem $GRBRACNNC-2D$ can be solved optimally in $O(n^4(m+n)^2)$ time using $O(m+n)$ space.
\end{theo}
%\begin{proof}
  % We have shown the proof of correctness of our algorithm above. Now we show the time complexity analysis as follows. The searching of two points $q_a$ and $q_b$ that lie on $TS({\cal R}_{in})$  and  $BS({\cal R}_{in})$ in our algorithm, respectively, and accordingly updating the penalty of each annulus that is generated, needs $O(m+n)$ time.  We can choose the six red points ($q_1$, $q_2$, $q_3$, $q_4$, $q_i$, $q_j$) in $n\choose 6$ ways, and hence it proves the result. 
%\end{proof}
\subsubsection{Non-uniform Concentric Annulus (\texorpdfstring{
$GRBRACNC-2D$}))}
\label{last}

Lemma~\ref{lm1} holds good for this problem.
We choose any four red points in $R$, and among these, two red points determine the horizontal width ($w_h$)  and the other two define the vertical width ($w_v$). Now, we can construct a non-uniform concentric annulus $\cal A$ by taking any other two points of any color that lie on the boundary of $\cal A$ (see Lemma~\ref{lm1}). We follow almost the same procedure as described in Section~\ref{non_uni_con} to obtain the following result.
\begin{theo}
The problem $GRBRACNC-2D$ can be solved in $O(n^4(m+n)^2)$ time and $O(m+n)$ space.
\end{theo}
\begin{proof}
    We choose the four red points in $n\choose 4$ (i.e., $O(n^4)$ ) ways to determine a pair of widths (say $Left(\cal A)$ and $Bot(\cal A)$) of all possible annuli. Now for each such pair, we select two points of any color in $R\cup B$ that defines two of the remaining sides of the rectangles in ${\cal A}$, say $TS({\cal R}_{out})$ and $RS({\cal R}_{out})$ and construct a non-uniform concentric annulus $\cal A$. We can shift $Top(\cal A)$ upwards, keeping its width same, and compute its penalty for each different $Top(\cal A)$ in $O(m+n)$ times. Then we shift $Right(\cal A)$ at most ($m+n-4$) times and repeat the above steps. For each of the fixed positions of $Left(\cal A)$ and $Bot(\cal A)$ (that are determined by the two red points on each of these regions), computations of all possible annuli along with their penalties need $O(m+n)^2$ time. This proves our result. 
\end{proof}

\subsubsection{Uniform Annulus (\texorpdfstring{$GRBRACU-2D$}))}

Lemma~\ref{lm2} holds good for this problem.
We choose any two red points in $R$ that determine the width $w$ of the annulus. Now, we can construct a uniform annulus $\cal A$ by taking three points of any color that lie on the boundaries of $\cal A$ (see Lemma~\ref{lm2}). We follow almost the same procedure as in Section~\ref{uni-2d} to obtain the following result.

\begin{theo}
The problem $GRBRACU-2D$ can be  solved optimally in $O(n^2(m+n)^3)$ time and $O(m+n)$ space.
\end{theo}
\begin{proof}
    We can choose any two red points in $R$ to determine the width $w$ of the uniform annulus in $n \choose 2$, i.e., $O(n^2)$ ways. These two points define one of the regions of $\cal A$, say $Left(\cal A)$. We can define the other three regions of $\cal A$ by taking three points, each one of any color. We shift one of these regions keeping its width same, and generate all different annuli along with their penalties. For a fixed position of $Left(\cal A)$, the above steps need $O(m+n)^3$ time. We need linear space to store the points. This proves the result.
\end{proof}
\subsection{Restricted version of the 
Generalized Red-blue rectangular annulus cover problem}
\label{joco5}

In this section, we consider the restricted version of $GRBRAC$ problems in two dimensions, where 
both the rectangles ${\cal R}_{out}$ and ${\cal R}_{in}$ of the axis-parallel rectangular annulus $\cal A$, are centered on a given horizontal line $L$. The annulus with this constraint is said to be a restricted annulus. Throughout this section, we use restricted annulus to imply a restricted rectangular annulus unless otherwise stated.

\begin{obsv}
    \label{same_width}
    In a restricted annulus ${\cal A}$, the distance of its $TS({\cal R}_{out})$ (resp. $TS({\cal R}_{in}))$ and $BS({\cal R}_{out})$ (resp. $BS({\cal R}_{in}))$ from the line $L$ are same since the center of the rectangle ${\cal R}_{out}$ (resp. ${\cal R}_{in}$) lies on $L$. The top and the bottom width of the restricted annulus ${\cal A}$ must be same. The restricted annulus $\cal A$ is symmetric about the horizontal line $L$.
\end{obsv}
Observation~\ref{same_width} says that the following three types of the {\color{blue}restricted annulus} $\cal A$ is possible based on its widths.
\begin{itemize}
\setlength\itemsep{0.5em}
    \item[(i)] {\color{blue} The uniform annulus}: All the four widths of $\cal A$ are the same (say, $w$), as shown in Figure~\ref{linerec} (a).
    \item[(ii)] {\color{blue}  The non-uniform non-concentric annulus}: The top and bottom widths are equal (say, $w_1$), whereas the left width $w_2$ and right width $w_3$ are not equal, i.e., $w_2 \neq w_3$ (see Figure~\ref{linerec}   (b)).
    \item[(iii)] {\color{blue}  The non-uniform concentric annulus}: The top and bottom widths are the same (say, $w_1$), and also, the left and right widths are the same (say, $w_2$); however, $w_1 \neq w_2$ (see Figure~\ref{linerec} (c)).
\end{itemize}
%We observe the following properties of the restricted version of three different types of annulus.

\noindent Based on the aforesaid three types, i.e., uniform, non-uniform concentric, and non-uniform non-concentric of the restricted annulus, we study the corresponding restricted versions of $GRBRAC$ problem which are referred to as $Restricted-GRBRACU$, $Restricted-GRBRACNC$ and $Restricted-GRBRACNNC$, respectively.
We present the algorithms for all these problems to determine their optimal solutions in the following three subsections.
\subsubsection{Uniform Annulus (\boldmath \texorpdfstring{$Restricted-GRBRACU$}))}
\label{restricted1}
In this section, we compute the uniform rectangular annulus with a minimum penalty whose center lies on a given horizontal line $L$.
\begin{lema}
    \label{rru_2}
    Four points are sufficient to define the {\color{blue} restricted version of a uniform rectangular annulus}, and among these points, two must be red
    so that the annulus has a minimum penalty.
\end{lema}
\begin{proof}
    The width $w$ of a uniform annulus ${\cal A}$ is determined by two red points lying on the same side of ${\cal R}_{out}$ and ${\cal R}_{in}$; otherwise, it can be reduced further to minimize the penalty of $\cal A$.
    If the width $w$ of $\cal A$ is determined by  $Left(\cal A)$ (resp. $Right(\cal A)$), then a single point of any color (red or blue) is sufficient to define both $Top(\cal A)$ and $Bot(\cal A)$ of width $w$, since $\cal A$ is symmetric about a given horizontal line $L$ (see Observation~\ref{same_width}); and also $Right(\cal A)$ (resp. $Left(\cal A)$) is defined by a single point of any color.

     Similarly, if the width $w$ of $\cal A$ is determined by $Top(\cal A)$ (resp. $Bot(\cal A)$), then two points, each one of any color (red or blue) is sufficient to define $Left(\cal A)$ and $Right(\cal A)$  of the same width (=$w$) separately.  This proves the result. 
\end{proof}

\noindent\textbf{Algorithm}:
Without loss of generality, we assume that the width $w$ of the restricted uniform annulus is determined by the two red points, say $q_i$ and $q_j$, lying on the left sides of ${\cal R}_{out}$ and ${\cal R}_{in}$, respectively, i.e., $w=|x(q_i)-x(q_j)|$.
\begin{figure}[ht]
    \centering
    \includegraphics[width=\textwidth]{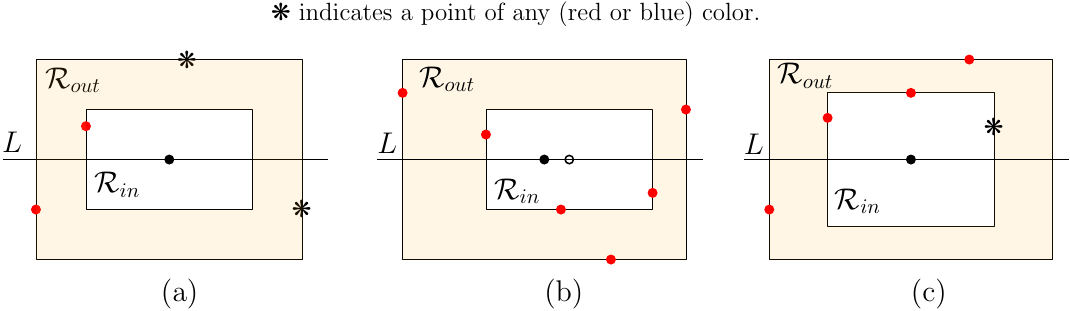}
    \caption{(a) Uniform annulus, (b) Non-uniform non-concentric annulus, and (c) Non-uniform concentric annulus. The centers of the rectangles are restricted to lie on $L$.}
    \label{linerec}
\end{figure}
We choose a point, say $s$, that lie to the right of $q_j$, where the color of $s$ may be red or blue, and define $right({\cal R}_{out})$. Let $V_1$, $V_2$ and $V_4$ be the vertical lines passing through the points $q_i$, $q_j$, and $s$, respectively. We take another vertical line $V_3$ to the left of $V_4$, so that the horizontal distance between $V_3$ and $V_4$ is $w$. Now, we compute all possible uniform annuli for which $Top(\cal A)$ and $Bot({\cal A})$ are bounded by the pair $(V_1,~V_4)$. 

First, we construct a uniform annulus $\cal A$ of width $w$ with one of the points $q_i$, $q_j$, and $s$ being at any one of the eight corners of  $\cal A$, so that the distance of both $Top(\cal A)$ and $Bot(\cal A)$ from the line $L$ are equal. We call this annulus as a basic uniform annulus (denoted by ${\cal A}_{basic}$) and compute its penalty. 
Without changing the position of $Left(\cal A)$ and $Right(\cal A)$, we construct all possible uniform annuli of width $w$ as follows.

Keeping the width of ${\cal A}_{basic}$ constant (=$w$), we shift both $Top(\cal A)$ and $Bot(\cal A)$ in upwards and downwards direction by the same distance from $L$, sequentially through the
sorted points (in $R\cup B$) lying above and below $L$, so that either $Top(\cal A)$ or $Bot(\cal A)$ is defined by one of these points, and the annuli are generated sequentially. Once we generate a new annulus, we obtain its penalty by updating the penalty of the previous annulus by adding or subtracting the penalty of the point (depending on its color) that enters or leaves the new annulus. In this way, we generate all possible annuli whose $Left(\cal A)$ is defined by the points $q_i$ and $q_j$, and $Right(\cal A)$ is defined by $s$. We keep track of the annulus with minimum penalty among all the annuli generated so far. Next, keeping the width of $\cal A$ constant (=$w$), we shift $Right(\cal A)$ rightwards sequentially through all the sorted points (with respect to their $x$-coordinates) and for each such $Right(\cal A)$, we repeat the above procedure to generate all the annuli of width $w$ defined by $q_i$ and $q_j$. We update the annulus with a minimum penalty if necessary. 
We choose all possible  $n\choose 2$ pairs of red points ($q_i,~q_j$) as the defining points for the left width $w$ of the uniform annulus and repeat the above procedure.

Similarly, we can execute the above algorithms to determine the annulus whose $Right(\cal A)$ defines the width of $\cal A$.

Now the width of a uniform annulus can also be determined by two red points lying on the boundaries of $Top(\cal A)$ or $Bot(\cal A)$. We show how to generate the annuli with $Top(\cal A)$, which is defined by two red points, say $q_i$ and $q_j$. We choose two points $s$ and $t$ (with any color (blue or red)) and construct a uniform annulus of width, say $w'=|y(q_i)-y(q_j)|$ made by $q_i$, $q_j$, $s$ and $t$, where $s$ and $t$ define $Left(\cal A)$ and $Right(\cal A)$, respectively. We shift only $Right(\cal A)$ rightwards and generate all possible annuli with their fixed position of $Left(\cal A)$. We can also shift $Left(\cal A)$ and repeat the same tasks (mentioned above) to generate all possible annuli for which $Top(\cal A)$ (resp. $Bot(\cal A)$) is at a fixed position. Finally, we choose all possible pairs of red points to define $Top(\cal A)$ (resp. $Bot(\cal A)$) and obtain the annulus ${\cal A}_{\min}$ with the minimum penalty.\\ 

\noindent{\bf Proof of correctness:} The correctness proof of our algorithm is based on the following claim.
\begin{clam}
    Our algorithm generates all possible uniform annuli.
\end{clam}
\begin{proof}
    We consider all the pairs ($q_i$, $q_j$) of red points in $R$ to define $Left(\cal A)$ (resp. $Right(\cal A)$, $Top({\cal A})$ or $Bot(\cal A)$) of the annulus $\cal A$, which determines all possible widths of the annuli to be computed. Then we generate all possible uniform annuli $\cal A$ satisfying the property in Lemma~\ref{rru_2} of a restricted uniform annulus whose center always lies on the line $L$. This proves our result. 
\end{proof}
Among all such annuli, we report the one with a minimum penalty, and hence this is an optimal algorithm.
\begin{theo}
    Our algorithm provides the optimal solution to $Restricted-GRBRACU$ problem in $O(n^2(m+n)^2)$ time and $O(m+n)$ space.
\end{theo}
\begin{proof}
    We have already proved that our algorithm is correct. Now we discuss its time complexity analysis.  Once we fix the position of $Left(\cal A)$ and $Right(\cal A)$ having width $w$ (which is determined as described in the proof of Lemma~\ref{rru_2}), we generate all possible corresponding $Top(\cal A)$ and $Bot(\cal A)$ of width $w$ that are at the same distance away from the line $L$, and it requires $O(m+n)$ time. Next, shift $Right(\cal A)$ rightward at most ($m+n$) times so that the right boundaries of ${\cal R}_{out}$ or ${\cal R}_{in}$ passes through the bichromatic points and repeat the same tasks. Finally, we choose the annulus with the minimum penalty. Since there are $O(n^2)$ pairs of red points, our algorithm needs $O(n^2(m+n)^2)$ time, and hence it proves the result.  
\end{proof}

\subsubsection{Non-uniform Concentric Annulus (\boldmath \texorpdfstring{$Restricted-RBRACNC$}~)}
In this section, we compute the non-uniform concentric rectangular annulus whose center lies on a given horizontal line $L$.

\begin{lema}
    \label{rru_3}
      Five points are sufficient to define the {\color{blue} restricted version of a non-uniform concentric rectangular annulus}, and among these points, four must be red so that the annulus has a minimum penalty.
\end{lema}
\begin{proof}
 Both the horizontal and the vertical width of the non-uniform concentric annulus ${\cal A}$ must be determined by two red points each; otherwise, these widths can be minimized further to obtain another annulus with less penalty. To define the horizontal width $w_h$ of $\cal A$, two red points must lie on $LS({\cal R}_{out})$ and $LS({\cal R}_{in})$ or on $RS({\cal R}_{out})$ and $RS({\cal R}_{in})$. Similarly, for the vertical width $w_v$ of ${\cal A}$ to be defined, the two red points lie on the top or bottom boundaries of ${\cal R}_{out}$ and ${\cal R}_{in}$. If these two red points determine the left (resp. right) width, then a single point of any color lying on the right (resp. left) boundaries of ${\cal R}_{out}$ or ${\cal R}_{in}$, is needed to determine the corresponding right (resp. left) width. So, three points must lie on any three vertical sides of ${\cal R}_{out}$ and ${\cal R}_{in}$. The annulus $\cal A$ is symmetric about the horizontal line $L$ (see Observation~\ref{same_width}). So, two points are sufficient to define both $Top({\cal A})$ and $Bot({\cal A})$, and they must be red as mentioned above. This proves the result. 
\end{proof}

\noindent The algorithm for generating annulus in this case is almost similar to that of the problem in Section~\ref{restricted1}. The only difference is that $Top(\cal A)$ or $Bot(\cal A)$ is defined by two red points (see Lemma~\ref{rru_3}). Hence in this case, we compute all the annuli by increasing the width of $Top(\cal A)$ and $Bot(\cal A)$ monotonically by sequentially traversing through all the red points lying above and below the line $L$.

\begin{theo}
    Our algorithm solves $Restricted-GRBRACNC$ problem optimally in $O(n^3(m+n)^2)$ time and $O(m+n)$ space.
\end{theo}
\begin{proof}
    Similar to the problem in Section~\ref{restricted1}, our algorithm generates all possible non-uniform concentric annuli $\cal A$ whose center always lies on the line $L$, and finally, we choose the annulus with the minimum penalty. It needs $O(n(m+n))$ time to generate as well as compute the penalty of all the annuli with its $Left(\cal A)$ and $Right(\cal A)$ being at a fixed position. Then keeping the width same, we shift $Right(\cal A)$ at most $(m+n)$ times (as in the algorithm in Section~\ref{restricted1}) and repeat the same tasks. Since there are $O(n^2)$ pairs of red points to determine the horizontal width, our algorithm needs $O(n^3(m+n)^2)$ time, and thus it proves the result. 
\end{proof}

\subsubsection{Non-uniform Non-Concentric Annulus (\boldmath \texorpdfstring{$Restricted-GRBRACNNC$}~)}
In this section, we compute the non-uniform non-concentric rectangular annulus whose centers (see Figure~\ref{linerec} (b)) lie on a given horizontal line $L$. 

\begin{lema}
    \label{rru_4}
      Six red points are sufficient to define the {\color{blue} restricted version of a non-uniform non-concentric rectangular annulus} so that the annulus has a minimum penalty. 
\end{lema}
\begin{proof}
    In a non-uniform non-concentric annulus $\cal A$, the widths $w_\ell$ and $w_r$ are different, whereas $w_t$ and $w_b$ are equal. To define $w_\ell$ (resp. $w_r$), we need two red points on the left (resp. right) sides of both ${\cal R}_{out}$ and ${\cal R}_{in}$, otherwise the width can be minimized further to reduce the penalty of ${\cal A}$. The annulus $\cal A$ is symmetric about the horizontal line $L$ (see Observation~\ref{same_width}). So, two red points are sufficient to define both $Top({\cal A})$ and $Bot({\cal A})$. Therefore, a total of six red points are sufficient to uniquely define such a restricted annulus $\cal A$ with minimum penalty. 
\end{proof}
In this case, the algorithm for generating annulus is almost similar to that of the problem in Section~\ref{restricted1}. The only difference is that both $Left(\cal A)$ and $Right(\cal A)$ are defined by two red points each (see Lemma~\ref{rru_4}). We compute all the annuli by increasing the width of $Top(\cal A)$ and $Bot(\cal A)$ monotonically by traversing sequentially through all the red points lying above and below the line $L$, respectively.
\begin{theo}
    Our algorithm generates the correct result for the problem $Restricted-GRBRACNNC$ in $O(n^5(m+n))$ time and $O(m+n)$ space.
\end{theo}
\begin{proof}
    Similar to the problem in Section~\ref{restricted1}, our algorithm generates all $feasible$ non-uniform non-concentric annuli $\cal A$ centered on the line $L$, and finally, we select the annulus with the minimum penalty. The time required to generate all the annuli with its $Left(\cal A)$ and $Right(\cal A)$ being at a fixed position needs $O(n(m+n))$ time. We can choose four different red points to define $Left(\cal A)$ and $Right(\cal A)$ in $n\choose 4$ ways. In other words, there are $O(n^4)$ choices of different left and right widths. For the fixed positions of $Left(\cal A)$ and $Right(\cal A)$, we can shift $Top(\cal A)$ and $Bot(\cal A)$ to compute all possible annuli $\cal A$ in $O(n(m+n))$ time. Hence, our algorithm needs $O(n^5(m+n))$ time, and thus it proves the result. 
\end{proof}

\section{Circular annulus}
\label{joco6}
In this section, we compute a circular annulus in $RBAC$ and $GRBAC$ problems, i.e., we consider $RBCAC$ (Red-blue circular annulus cover) and $GRBCAC$ (Generalized Red-blue circular annulus cover) problems.
Throughout this section, we use the term annulus to imply a circular annulus unless otherwise stated. In this section, we use the notations ${\cal C}_{out}({\cal A})$, ${\cal C}_{in}(\cal A)$, $VD(R)$, $FVD(R)$, $VD_i$, $VD_{i,j}$, $FVD_i$, and $FVD_{i,j}$ with their proper meaning as mentioned in Section~\ref{prelim}.
The points lying on the boundary of a circle $\cal C$ are said to be the {\color{blue} defining points} for $\cal C$.
 A circular annulus has the following property~\cite{BergCKO08}.
\begin{prop}\cite{BergCKO08}
\label{default}
    Four points are sufficient to define a circular annulus $\cal A$, where either (i) three points lie on the boundary of its outer circle  ${\cal C}_{out}(\cal A)$ and one point lies on its inner circle ${\cal C}_{in}(\cal A)$, (ii) three points lie on ${\cal C}_{in}(\cal A)$ and one point lies on ${\cal C}_{out}(\cal A)$, or (iii)~two points lie on ${\cal C}_{in}(\cal A)$ and two points lie on ${\cal C}_{out}(\cal A)$. 
\end{prop}
%{\color{blue} Note that in a special case, one of the circles (outer or inner) may be defined by two diametrically opposite points, and the other one is defined by a single point.}

\noindent We also state the following properties of $FVD(R)$ and $VD(R)$~\cite{BergCKO08,books/sp/PreparataS85}.
\begin{prop}\cite{BergCKO08}
\label{chap4prop1}
The centers of the circle passing through the three points in $R$ and containing no (resp. all) other points in $R$ define
the vertices of the Voronoi diagram $VD(R)~($resp. farthest-point Voronoi diagram $FVD(R))$.
\end{prop}

\begin{prop}\cite{BergCKO08}
\label{chap4prop2}
The locus of the center of the largest empty circles passing through only a pair of points ($q_i$, $q_j$), where $q_i,~q_j\in R$,
defines an edge of $VD(R)$, whereas the edge in a $FVD(R)$ is defined similarly with a difference that the circles cover all the red points in $R$.
\end{prop}

\begin{prop}\cite{BergCKO08,books/sp/PreparataS85}
\label{update}
    Given a Voronoi diagram $VD(R)$ with $|R|=n$. We can insert a new point $a$ (resp. two new points $a$ and $b$) and update $VD(R)$ to obtain $VD(R \cup \{a\})$ (resp. $VD(R \cup \{a,b\})$ ) in $O(n)$ time.
Similarly, we can update a given farthest-point Voronoi diagram $FVD(R)$ to obtain $FVD(R \cup \{a\})$ and $FVD(R \cup \{a,b\})$. 
\end{prop} 

\noindent Property~\ref{update} says that we can update $VD(R)$ (resp. $FVD(R)$) to obtain $VD_i$ and $VD_{i,j}$ (resp. $FVD_i$ and $FVD_{i,j}$) in linear time.
We first discuss $RBCAC$ problem, followed by its generalized version $GRBCAC$ problem (where the penalties are assigned to each point).

\subsection{Red-Blue Circular annulus cover  
%\texorpdfstring{($RBCAC$)}~
problem}
\label{joco7}
A circular annulus ${\cal A}$ is said to be $feasible$ if it covers all the red points in $R$.
The objective of Red-Blue Circular annulus cover problem (referred to as $RBCAC$ problem) is to compute a $feasible$ circular annulus covering a minimum number of blue points in $B$. Note that both $RBCAC$ and $RBRAC$ problems have the same result in $1$-dimension. So, we mainly discuss $RBCAC$ problem in $2$-dimension. We observe the following.
\begin{obsv}
\label{single_red}
    The colors of all the points lying on ${\cal C}_{out}(\cal A)$ $($resp. ${\cal C}_{in}(\cal A))$ cannot be blue in an optimal solution.
   \end{obsv}
\noindent {\bf Justification of Observation~\ref{single_red}:}  If the colors of all the points lying on ${\cal C}_{out}(\cal A)$ $($resp. ${\cal C}_{in}(\cal A))$ are blue, then keeping the center of ${\cal C}_{out}(\cal A)$ $($resp. ${\cal C}_{in}(\cal A))$ fixed at the same position, we can decrease (resp. increase) the radius of it to obtain another annulus containing less number of blue points. Hence, at least one of the point lying on  ${\cal C}_{out}(\cal A)$ $($resp. ${\cal C}_{in}(\cal A))$ must be red. \hfill $\blacksquare$

\noindent The distance of a point $p$ from a circle ${\cal C}$ is given by $|(|pc|-r)|$, where $c$ and $r$ are the center and radius of $\cal C$, respectively.
\begin{obsv}
 \label{shift}
    If the centers of two intersecting circles ${\cal C}_1$ and ${\cal C}_2$ are separated by a distance $\epsilon$, then the distance of a point on ${\cal C}_2$ from ${\cal C}_1$ is at most $2\epsilon$.
\end{obsv}
\noindent {\bf Justification of Observation~\ref{shift}:} Let $c_1$ and $r$ (resp. $c_2$ and ${\mathscr R}$) be the center and radius of the circle ${\cal C}_1$ (resp. ${\cal C}_2$), respectively (see Figure~\ref{epsilon}). Assume that $e$ is a point of intersection of the circles ${\cal C}_1$ and ${\cal C}_2$. In $\triangle c_1c_2e$, ${\mathscr R} \leq r+\epsilon$ (by triangle inequality). Hence, $|gh|={\mathscr R}-|c_2g|={\mathscr R}-(r-\epsilon) \leq (r+\epsilon)- (r-\epsilon)=2\epsilon$. \hfill $\blacksquare$
    \begin{figure}[ht]%{r}{0.4\textwidth}
 %\vspace{-0.2in}
 \centering
 \includegraphics[scale=0.85]{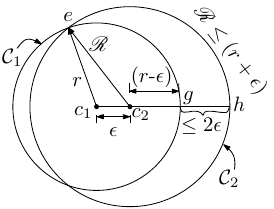}
  \caption{ An instance of Observation~\ref{shift}.}
  \vspace{-0.1in}
    \label{epsilon}
   \end{figure} 

\begin{defn}[$\delta({\cal A})$]
\label{delta}
 Compute the distances of all the points $a \in (R \cup B)$ from ${\cal C}_{out} ({\cal A})$ as well as from ${\cal C}_{in}({\cal A})$, except the point ${a}$ that lies on ${\cal C}_{out}({\cal A})$ or ${\cal C}_{in}({\cal A}))$. Among all such distances computed, the smallest one is defined as {$\delta({\cal A})$}. Thus $\delta({\cal A})$ is determined by both ${\cal C}_{out}({\cal A})$ and ${\cal C}_{in}({\cal A})$.
\end{defn}
\begin{defn}[$ohp(a, VD(\{a,b\}))$]
\label{hp}
    For any two points $a$ and $b$, we define {$ohp(a,$} {$VD(\{a,b\}))$} as the {\em open half-plane} of the partitioning line of $a$ and $b$ in $VD(\{a,b\})$ that contains the point $a$.
\end{defn}
\begin{lema}
\label{one_blue}
    Let ${\cal C}_{out}(\cal A)~($resp. ${\cal C}_{in}(\cal A))$ be defined by three points. If one (resp. two) of these points, say $p_i$ (resp. $p_i$ and $p_j$), is (resp. are) blue, then we can shift the center of ${\cal A}$ and change its width to obtain an annulus, say ${\cal A}'$, which covers the same set of points that are covered by $\cal A$ except the blue point(s) $p_i$ (resp. $p_i$ and $p_j$).   
\end{lema}
\begin{proof}
The following two cases need to be considered (see Observation~\ref{single_red}). 
\begin{description}
    \item[(i)] One of the three defining points of a circle is blue:\\
    Without loss of generality, we assume that ${\cal C}_{out}(\cal A)$ is defined by three points, one of which, say $p_i$, is blue, and the other two, say $q_j$ and $q_k$, are red (see Figure~\ref{2r1b}~(a)). In this case, ${\cal C}_{in}(\cal A)$ is defined by a single point (see Property~\ref{default}) of red color (see Observation~\ref{single_red}), and we take this point, say $q_\ell\in R$. 
We compute $\delta(\cal A)$ for the annulus $\cal A$ (see Definition~\ref{delta}).
We shift the center of $\cal A$ towards the midpoint of $\overline{q_jq_k}$ by a small distance $\epsilon<\frac{{\delta}(\cal A)}{2}$, and create a new annulus ${\cal A}'$ centered at, say $c'$, with $|c'q_j|$ and $|c'q_\ell|$ as the radii of ${\cal C}_{out}({\cal A}')$ and ${\cal C}_{in}({\cal A}')$, respectively (see Figure~\ref{2r1b} (b)). Since $\epsilon<\frac{{\delta}(\cal A)}{2}$, the annulus $\cal A'$ covers exactly the same set of points that are covered by $\cal A$, except the blue point $p_i$ (see Observation~\ref{shift}). Hence, this proves our result if one of the three defining points of ${\cal C}_{out}(\cal A)$ is a blue point.

Similarly, we can prove the result if ${\cal C}_{in}(\cal A)$ is defined by three points: one blue point and two red points.

    \item[(ii)] Two of the three defining points of a circle are blue:\\
   Without loss of generality, we assume that ${\cal C}_{out}(\cal A)$ is defined by three points, one of which, say $q_k\in R$ is red, and the other two, say $p_i,~p_j\in B$, are blue (see Figure~\ref{2b1r}~(a)). Now, ${\cal C}_{in}(\cal A)$ must be defined by a single red point, say $q_\ell$ (see Property~\ref{default} and Observation~\ref{single_red}). We compute $\delta(\cal A)$ for the annulus $\cal A$ (see Definition~\ref{delta}). We shift the center of $\cal A$ towards the red point $q_k$ along $\overline{cq_k}$ by a distance $\epsilon<\frac{{\delta}(\cal A)}{2}$ to create  a new annulus ${\cal A}'$ centered at, say $c'$, with $|c'q_k|$ and $|c'q_\ell|$ as the radii of ${\cal C}_{out}({\cal A}')$ and ${\cal C}_{in}({\cal A}')$, respectively (see Figure~\ref{2b1r}~(b)). Since $\epsilon<\frac{{\delta}(\cal A)}{2}$, the annulus ${\cal A'}$ covers all the points inside ${\cal A}$ except the blue points $p_i$ and $p_j$ (see Observation~\ref{shift}). This proves the result.

Similarly, we can also show that the result is correct if ${\cal C}_{in}(\cal A)$ is defined by three points: two blue points and one red point. 
\end{description} 
\end{proof}

   \begin{figure}[ht]
 \centering
 \includegraphics[scale=0.85]{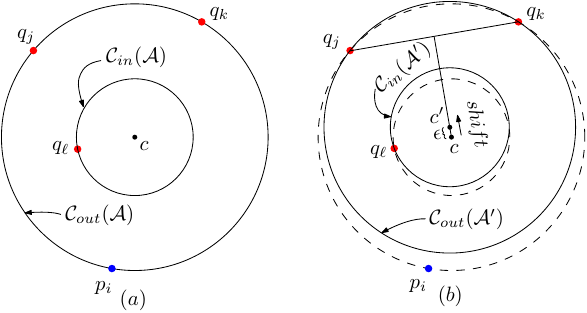}
  \caption{Two red points ($q_j$, $q_k$) and one blue point ($p_i$) define ${\cal C}_{out}(\cal A)$.}
  \vspace{-0.1in}
    \label{2r1b}
   \end{figure}

\begin{figure}[bt]%{r}{0.4\textwidth}
 \centering
 \includegraphics[scale=0.85]{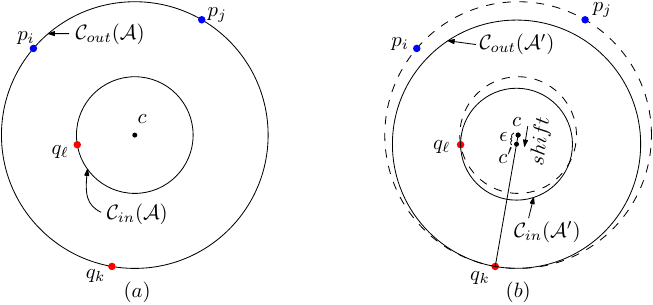}
  \caption{Two blue points ($p_i$, $p_j$) and one red point ($q_k$) define ${\cal C}_{out}(\cal A)$.}
  \vspace{-0.1in}
    \label{2b1r}
   \end{figure}

   \begin{lema}
\label{two_blue_1}
    If both ${\cal C}_{out}(\cal A)$ and ${\cal C}_{in}(\cal A)$ are defined by two points each, and one of these four defining points, say $p_i$ is blue, then we can shift the center of $\cal A$ and change its width to obtain another annulus, say $\cal A'$, which covers the same set of points that are covered by $\cal A$ except the blue point $p_i$.   
\end{lema}
\begin{proof}
    It is similar to Case (ii) in the proof of Lemma~\ref{one_blue}. 
\end{proof}
   
\begin{lema}
\label{two_blue}
    If both ${\cal C}_{out}(\cal A)$ and ${\cal C}_{in}(\cal A)$ are defined by two points each, and one of the defining points of each circle, say $p_i$ (resp. $p_j$) on ${\cal C}_{out}(\cal A)~($resp. ${\cal C}_{in}(\cal A))$ is blue, then we can shift the center of $\cal A$ and change its width to obtain another annulus, say $\cal A'$, that contains the same set of points that are covered by ${\cal A}$ except the blue point(s) $p_i$ and/or $p_j$. In other words, we can generate an annulus $\cal A'$ containing at least one blue point less than that of $\cal A$.  
\end{lema}
\begin{proof}
    Let the red point that lies on ${\cal C}_{out}(\cal A)$ (resp. ${\cal C}_{in}(\cal A)$) be $q_i$ (resp. $q_j$). The perpendicular bisector of $\overline{p_iq_i}$ and $\overline{p_jq_j}$ intersect with each other at the center $c$ of the annulus $\cal A$. 
    We determine $\delta(\cal A)$ for the annulus ${\cal A}$ (see Definition~\ref{delta}).
    There are three possible cases depending on the positions of $q_i$ and $q_j$ with respect to the perpendicular bisector of $\overline{p_jq_j}$, where the two points $p_j$ and $q_j$ lies on ${\cal C}_{in}(\cal A)$.
    \begin{itemize}
    \setlength\itemsep{0.5em}
        \item[(i)] The two red points $q_i$ and $q_j$ lie on the two opposite sides of a bisector $\overline{p_jq_j}$. We shift the center $c$ towards the red point $q_i$ (see Figure~\ref{fig:same_doublept}) by a distance $\epsilon<\frac{\delta(\cal A)}{2}$ and create another new annulus ${\cal A'}$, centered at, say $c'$ with $|c'q_i|$ and $|c'q_j|$ as the radii of ${\cal C}_{out}(\cal A')$ and ${\cal C}_{in}(\cal A')$, respectively. In this case, all the points in $\cal A$, except the blue points $p_i$ and $p_j$, remain covered by $\cal A'$. 
        \item[(ii)] The two red points $q_i$ and $q_j$ lie on the same side of the bisector $\overline{p_jq_j}$.\\
        Depending on whether the same line or two different lines are the bisectors of both the pairs ($p_i$, $q_i$) and ($p_j$, $q_j$), we have the following two cases:        
        \begin{itemize}
        \setlength\itemsep{0.5em}
        \item[(a)]  The perpendicular bisectors of the pairs ($p_i$, $q_i$) and ($p_j$, $q_j$) are different lines:\\
        In this case (see Figure~\ref{fig:diff_doublept}),
        we consider an arbitrary point, say $c'$ at a distance $\epsilon<\frac{\delta(\cal A)}{2}$ (see Definition~\ref{hp}) from $c$ inside the region $ohp(p_j,VD(\{p_j,q_j\}))\cap ohp(q_i,VD($ $\{q_i,p_i\}))$ (see the shaded region in Figure~\ref{fig:diff_doublept}(a)), and create a new annulus ${\cal A}'$ centered at $c'$ with $|c'q_i|$ and $|c'q_j|$ as the radii of ${\cal C}_{out}(\cal A')$ and ${\cal C}_{in}(\cal A')$, respectively. The annulus ${\cal A}'$ covers exactly the same set of points that are covered by ${\cal A}$, except the blue points $p_i$ and $p_j$.
        \item[(b)] The perpendicular bisectors of the pairs ($p_i$, $q_i$) and ($p_j$, $q_j$) are the same line:\\ In this case (see Figure~\ref{fig:same_bisector}),
        we compute $\delta(\cal A)$ for ${\cal A}$ (see Definition~\ref{delta}). Then we choose a point $c'$ on $\overline{cq_i}$ at a distance $\epsilon \leq \frac{\delta(\cal A)}{2}$ from $c$ and construct an annulus ${\cal A}'$ centered at $c'$ with $|c'q_i|$ and $|c'q_j|$ as the radii of ${\cal C}_{out}(\cal A')$ and ${\cal C}_{in}(\cal A')$, respectively  (see Figure~\ref{fig:same_bisector}). In this case, all the blue points inside ${\cal A}$, except the blue point $p_i$, are covered by ${\cal A}'$. 
        \vspace{-0.5cm}
       \end{itemize}
      \end{itemize}
\end{proof}
\begin{figure}[ht]%{r}{0.4\textwidth}
 %\vspace{-0.2in}
 \centering
 \includegraphics[scale=1]{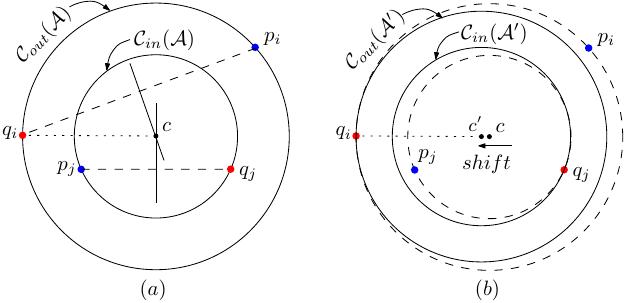}
  \caption{(a) Two red points $q_i$ and $q_j$ lie at the different sides of the perpendicular bisector, (b) After shifting of the center, both the blue points are removed from $\cal A'$.}
  \vspace{-0.1in}
    \label{fig:same_doublept}
   \end{figure} 
   \begin{figure}[ht]%{r}{0.4\textwidth}
 %\vspace{-0.2in}
 \centering
 \includegraphics[scale=1]{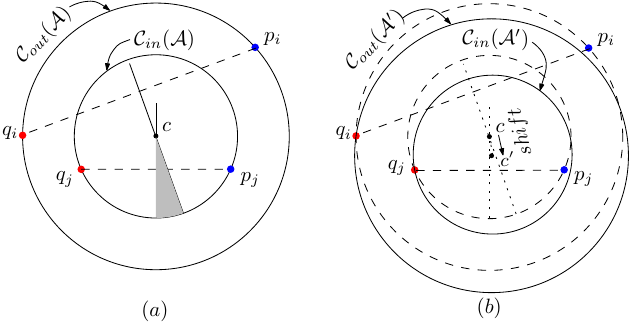}
  \caption{(a) Two red points $q_i$ and $q_j$ lie at the same side of the perpendicular bisector, (b) After shifting of the center, both the blue points are removed from $\cal A'$.}
  \vspace{-0.1in}
    \label{fig:diff_doublept}
   \end{figure} 
     \begin{figure}[ht]%{r}{0.4\textwidth}
 %\vspace{-0.2in}
 \centering
 \includegraphics[scale=1]{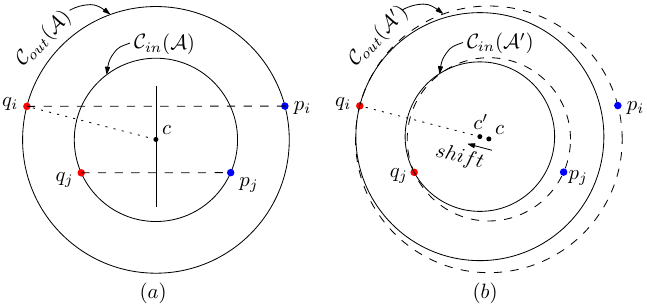}
  \caption{(a) Perpendicular bisectors of $\overline{q_ip_i}$ and $\overline{q_jp_j}$ are the same, and the red points $q_i$ and $q_j$ lie at the same side of the perpendicular bisector, (b) After shifting of the center, one blue point is removed from $\cal A'$.}
  \vspace{-0.1in}
    \label{fig:same_bisector}
   \end{figure} 
We describe our algorithm for $RBCAC$ problem as follows.\\[0.1in]
{\bf Algorithm:}
%\\[0.1in] 
As a preprocessing task, we compute the Voronoi diagram $VD(R)$ and the farthest-point Voronoi diagram $FVD(R)$ of the given set $R$ of red points. We use two doubly connected edge lists ($DCEL$) to store $VD(R)$ and $FVD(R)$ separately, which need a linear space~\cite{BergCKO08}.
Throughout our algorithm, we always maintain an annulus ${\cal A}_{\min}$, which covers the minimum number of blue points among the annuli generated so far. 

We generate all the $feasible$ annuli in the following three cases, each one satisfying one of the different conditions for an annulus to be defined (see Property~\ref{default}).
\vspace{2cm}
\begin{description}
    \item[\bf Case (A):] {\bf \boldmath ${\cal C}_{out}(\cal A)$ and ${\cal C}_{in}(\cal A)$ are defined by three points and one point lying on its boundary, respectively.}\\
   The following three sub-cases are possible depending on the color of the defining points of ${\cal C}_{out}(\cal A)$.
    \begin{description}
    \item[\bf Case (A.1):] {\bf \boldmath All the defining points lying on the boundary of ${\cal C}_{out}(\cal A)$ are red.}\\
    Repeat the following tasks at each vertex, say $c$, of $FVD(R)$ to compute all the $feasible$ annuli $\cal A$ with ${\cal C}_{out}(\cal A)$ and ${\cal C}_{in}(\cal A)$ passing through three red points and a single red point, respectively.\\
We choose $c$ as the center of ${\cal C}_{out}(\cal A)$, which passes through the three red points and covers all the red points in $R$ (see Property~\ref{chap4prop1}). Then we search $VD(R)$ to find out the Voronoi cell that contains the center $c$, and we compute the red point, say $q_i$, closest from $c$ which determines the radius of ${\cal C}_{in}(\cal A)$.  Since $\cal A$ covers all the red points in $R$, the annulus $\cal A$ is a $feasible$ one. Next, we compute the number of blue points lying inside $\cal A$.  If the number of blue points inside ${\cal A}$ is less than that of ${\cal A}_{\min}$, then we update ${\cal A}_{\min}$ by taking $\cal A$ as ${\cal A}_{\min}$.\\ 
    \item[ \bf Case (A.2):] {\bf \boldmath ${\cal C}_{out}(\cal A)$ has one blue point and two red points on its boundary.}\\
    For each blue point $p_i\in B$, we execute the following tasks.\\
We update $FVD(R)$  to obtain $FVD_{i}$, i.e., $FVD(R\cup\{p_i\})$ in linear time.
We use  $FR(i)$ to denote the region in $FVD_i$, which is farthest from the blue point $p_i$. We choose each of the vertices of this farthest region $FR(i)$ as a center $c$ of ${\cal C}_{out}(\cal A)$ which passes through $p_i$ along with two other red points and covers all the red points in $R$. Next,  we compute ${\cal C}_{in}(\cal A)$ and the number of blue points inside ${\cal A}$, in the same way as described in Case~(A.1).
Then we compute another annulus ${\cal A'}$ whose center is at a small distance $\epsilon>0$, away from that of $\cal A$, so that $\cal A'$ covers exactly the same set of points as covered by $\cal A$, except the blue point $p_i$ (see Lemma~\ref{one_blue}).
So, for the annulus $\cal A$, we generate a corresponding annulus $\cal A'$ by eliminating the blue point lying on ${\cal C}_{out}(\cal A)$. We update ${\cal A}_{\min}$, if $\cal A'$ covers less number of blue points than ${\cal A}_{\min}$. 
    \item[\bf Case (A.3):] {\bf \boldmath ${\cal C}_{out}(\cal A)$ has two blue points and one red point on its boundary.}\\
    In this case, for each pair of blue points ($p_i$, $p_j$) $\in B\times B$, we execute the following tasks.\\
 We update $FVD(R)$ to obtain $FVD_{i,j}$, i.e., $FVD(R\cup\{p_i,p_j\})$. We consider only those vertices of $FVD_{i,j}$ as the center $c$ of ${\cal C}_{out}(\cal A)$ which passes through two blue points $p_i$, $p_j$ and one red point, say $q_k\in R$. Then, similar to Case~(A.1), we compute ${\cal C}_{in}(\cal A)$ and the number of blue points lying inside $\cal A$. Next, we follow Lemma~\ref{one_blue} to construct another annulus $\cal A'$ as described in Case~(A.2), and update ${\cal A}_{\min}$, if $\cal A'$ covers less number of blue points than ${\cal A}_{\min}$. 
  \end{description}
 % \vspace{0.1in}
  \item[\bf Case (B):] {\bf \boldmath ${\cal C}_{out}(\cal A)$ has one point on its boundary and ${\cal C}_{in}(\cal A)$ has three points on its boundary.}\\
  Similar to Case~(A), it has the following three sub-cases depending on the color of the defining points of ${\cal C}_{in}(\cal A)$. 
  \begin{description}
      \item[ \bf Case (B.1):] {\bf \boldmath ${\cal C}_{in}(\cal A)$ has three red points on its boundary.}\\
      Repeat the following tasks at each vertex, say $c$, of $VD(R)$ to compute all the $feasible$ annuli $\cal A$ with ${\cal C}_{in}(\cal A)$ and ${\cal C}_{out}(\cal A)$ passing through three red points and a single red point, respectively.\\
      We choose $c$ as the center of ${\cal C}_{in}(\cal A)$, which passes through three red points and does not contain any other points in $R$ (see Property~\ref{chap4prop1}). Then we search in $FVD(R)$  to find the Voronoi site (a red point), say $q_k\in R$ so that the farthest region of $q_k$ in $FVD(R)$ contains the point $c$. The red point farthest from $c$ is $q_k$, which determines the radius of ${\cal C}_{out}(\cal A)$.  Since the annulus $\cal A$ covers all the red points in $R$, $\cal A$ is a $feasible$ annulus. We compute the number of blue points lying inside $\cal A$, and if necessary, we update ${\cal A}_{\min}$.     
      \item[ \bf Case (B.2):] {\bf \boldmath ${\cal C}_{in}(\cal A)$ has two red points and one blue point on its boundary.}\\ 
       In this case, for each blue point $p_i\in B$, we execute the following tasks.\\
       We update $VD(R)$ to obtain $VD_{i}$, i.e., $VD(R\cup\{p_i\})$, and let $VC(i)$ be the Voronoi cell of the site (blue point) $p_i$ in $VD_i$. We choose each vertex of the boundary of $VC(i)$ as a center $c$ of ${\cal C}_{in}(\cal A)$ which passes through $p_i$ and two other red points in $R$, and does not contain any other red points in $R$. Next,  we compute ${\cal C}_{out}(\cal A)$ and the number of blue points inside $\cal A$ in the same way as described in Case~(B.1).  Once we generate the annulus $\cal A$, we follow Lemma~\ref{one_blue} to construct another annulus ${\cal A}'$, similar to Case~(A.2) and then update ${\cal A}_{\min}$, if necessary.
      \item[ \bf Case (B.3):] {\bf \boldmath ${\cal C}_{in}(\cal A)$ has one red point and two blue points on its boundary.}\\ 
      In this case, we do the following tasks for each pair of blue points ($p_i$, $p_j$) $\in B\times B$.\\
We update $VD(R)$ to obtain $VD_{i,j}$, i.e., $VD(R\cup\{p_i,p_j\})$.
We consider only those vertices of $VD_{i,j}$ as the center $c$ of ${\cal C}_{in}(\cal A)$ which passes through two blue points $p_i$, $p_j$ and a red point, say $q_k\in R$. Then, we use $FVD(R)$ to find out the farthest red point from $c$, and thus we obtain ${\cal C}_{out}(\cal A)$. Then similar to Case~(A.2), we follow Lemma~\ref{one_blue} to construct another annulus ${\cal A}'$ and update ${\cal A}_{\min}$, if necessary. 
  \end{description}
 % \vspace{0.1in}
  \item[\bf Case (C):] {\bf \boldmath ${\cal C}_{out}(\cal A)$ and ${\cal C}_{in}(\cal A)$ have two points on their boundaries.}\\
    Depending on the color of the defining points of ${\cal C}_{in}(\cal A)$ and ${\cal C}_{out}(\cal A)$, we consider the following exhaustive four sub-cases.
  \begin{description}
      \item[ \bf Case (C.1):] {\bf \boldmath The two points on the boundaries of both ${\cal C}_{in}(\cal A)$ and ${\cal C}_{out}(\cal A)$ are red.}\\ 
      The center of the annulus $\cal A$ must lie on the edges of $VD(R)$ and $FVD(R)$ (see Property~\ref{chap4prop2}) so that ${\cal C}_{out}(\cal A)$ and ${\cal C}_{in}(\cal A)$ pass through two points. This implies that the center of all such annuli is the points of intersection of the edges of $VD(R)$ with that of $FVD(R)$. At each such point of intersection, we compute ${\cal C}_{in}(\cal A)$ and ${\cal C}_{out}(\cal A)$ which satisfy Property~\ref{chap4prop2} and thus $\cal A$ is a $feasible$ annulus. We count the number of blue points inside $\cal A$ and update ${\cal A}_{\min}$, if necessary.
      \item[ \bf Case (C.2):] {\bf \boldmath ${\cal C}_{out}(\cal A)$ has one red point and one blue point on its boundary, whereas both the points on ${\cal C}_{in}(\cal A)$ are red.}\\ 
      In this case, for each blue point $p_i\in B$, we execute the following.\\
  We update $FVD(R)$ to obtain $FVD_{i}$, i.e., $FVD(R\cup\{p_i\})$. We choose the farthest region (denoted by $FR(i)$) of the blue point $p_i$ in $FVD_i$ and find out the points of intersection of the boundary of $FR(i)$ with the edges of $VD(R)$. We choose each such point of intersection as the center $c$ of  ${\cal C}_{out}(\cal A)$, which passes through $p_i$ and a red point in $R$. We use $VD(R)$ to obtain an empty circle $\cal C$ centered at $c$ passing through two red points in $R$, and we take ${\cal C}$ as ${\cal C}_{in}(\cal A)$. Then we construct an annulus ${\cal A}'$ whose center is $\epsilon$ distance away from that of $\cal A$ (see Lemma~\ref{two_blue_1}) so that $\cal A'$ contains the same set of points that are covered by $\cal A$, except the blue point $p_i$. Then we count the number of blue points inside $\cal A'$, and if it is less than ${\cal A}_{\min}$, we update ${\cal A}_{\min}$.
      \item[ \bf Case (C.3):] {\bf \boldmath ${\cal C}_{in}(\cal A)$ has one red point and one blue point on its boundary, whereas both the points on ${\cal C}_{out}(\cal A)$ are red.}\\ 
      In this case, we execute the following tasks for each blue point $p_i\in B$.\\
We update $VD(R)$ to obtain $VD_{i}$, i.e., $VD(R\cup\{p_i\})$. In $VD_i$, we choose the Voronoi cell $VC(i)$ of the blue point $p_i$, whose boundary consists of at most $(n-1)$ edges. Next, we find out the point of intersection of the boundary of $VC(i)$ with the edges of $FVD(R)$. We choose each such point of intersections as the center $c$ of ${\cal C}_{in}(\cal A)$ that passes through $p_i$ and a red point, say $q_j\in R$. 
We use $FVD(R)$ to obtain a circle ${\cal C}_{out}(\cal A)$ centered at $c$, which covers all the red points in $R$ and passes through two red points in $R$. Then similar to Case~(C.2), we follow Lemma~\ref{two_blue_1} to construct another annulus ${\cal A}'$ and update ${\cal A}_{\min}$, if necessary.
      \item[ \bf Case (C.4):] {\bf \boldmath Both ${\cal C}_{out}(\cal A)$ and ${\cal C}_{in}(\cal A)$ have one red point and one blue point on their boundaries.}\\  
      In this case, for each pair of blue points ($p_i$, $p_j$)$\in B\times B$, we do the following.\\
Compute $VD_{i}$, i.e., $VD(R\cup\{p_i\})$ and $FVD_{j}$, i.e., $FVD(R\cup\{p_j\})$. Choose the Voronoi cell $VC(i)$ of the blue point $p_i$ in $VD_i$ that has at most $(n-1)$ edges along the boundary of $VC(i)$. Also consider $FR(j)$, the farthest region of the blue point $p_j$ in $FVD_j$. Then we find the points of intersections (at most $2(n-1)$ in number) between the boundaries of the cells $VC(i)$ and $FR(j)$. We choose these points of intersections as the center of the annulus. The circle ${\cal C}_{out}(\cal A)$ (resp. ${\cal C}_{in}(\cal A)$) passes through $p_j$ (resp. $p_i$) and a red point, say $q_j\in R$ (resp. $q_i\in R$). Then similar to Case~(C.2), we follow Lemma~\ref{two_blue} to construct another annulus ${\cal A}'$ and update ${\cal A}_{\min}$, if necessary.
  \end{description}
  Finally we report the annulus ${\cal A}_{\min}$.
\end{description}

\noindent {\bf Remarks:} Note that as a special case, one of the circles (outer or inner) may be defined by two diametrically opposite points, and the other one is defined by a single point. 
There are two possibilities as follows:
\begin{description}
   \item[(i)] ${\cal C}_{out}(\cal A)$ is defined by two diametric points:\\
To generate an annulus $\cal A$ with ${\cal C}_{out}(\cal A)$, we first compute the diameter (i.e., the farthest pair) of the point set $R$ from $FVD(R)$ and check whether all the points in $R$ can be covered by a circle centered at the midpoint of such a farthest pair. If so, then we take that circle as ${\cal C}_{out}(\cal A)$ and compute the corresponding ${\cal C}_{in}(\cal A)$ in a similar way as described above in the algorithm.
Similarly, we can also construct the annulus $\cal A$ with its ${\cal C}_{in}(\cal A)$ being defined by a single red point and a single blue point $p_i$ using $FVD_i$, as described above in our algorithm.
\item[(ii)]  ${\cal C}_{in}(\cal A)$ is defined by two diametric points:\\
We consider the two sites (red points), say $q_i,q_j\in R$ of the two adjacent Voronoi cells in $VD(R)$, and choose the point $c$ on the edge which is equidistant from this pair. We take $c$ as a center of $\cal A$ and then find out the red point farthest from $c$ using $FVD(R)$, and follow the same procedure in our algorithm to generate $\cal A$. Similarly, we handle the case, if one of the defining points of ${\cal C}_{in}(\cal A)$ is red.
\end{description}

\noindent{\bf Proof of correctness:} The correctness of the algorithm is based on the following claim.
\begin{clam}
    Our algorithm generates all the $feasible$ annuli.
\end{clam}
\begin{proof}
Our algorithm always generates all possible annuli $\cal A$ following Property~\ref{default} so that it covers all the red points in $R$. When one or two of the defining points of $\cal A$ are blue, we generate its corresponding annulus $\cal A'$ in such a way that $\cal A'$  discards such blue-colored defining points without eliminating any red ones. Hence, our algorithm generates all possible annuli which are $feasible$. 
\end{proof}
Since our algorithm chooses the $feasible$ annulus ${\cal A}_{\min}$ that covers the minimum number of blue points, it provides the optimal result.

\begin{theo}
  The algorithm for $RBCAC$ problem computes the optimal result in $O(m^2n(m+n))$ time and linear space.  
\end{theo}
\begin{proof}
 We have presented the algorithm along with its correctness proof above. Now we analyze the time complexity of the algorithm, which is mainly determined by the following tasks:
 \begin{itemize}
\setlength\itemsep{0.5em}
     \item[(i)] The preprocessing task, i.e., the computation of $VD(R)$ and $FVD(R)$ needs $O(n\log n)$ time.
     \item[(ii)] We update $VD(R)$ (resp. $FVD(R)$) to obtain the  Voronoi (resp. farthest-point Voronoi) diagram for the point sets $R\cup\{p_i\}$, $R\cup\{p_i,p_j\}$ in $O(n)$ time.
     \item[(iii)] In $VD(R)$ (resp. $FVD(R)$), we search for the Voronoi cell (resp. farthest region) in which the center $c$ of ${\cal C}_{out}(\cal A)$ (resp. ${\cal C}_{in}(\cal A)$) lies on, and it needs $O(\log n)$ time. Then we compute the corresponding ${\cal C}_{in}(\cal A)$ (resp. ${\cal C}_{out}(\cal A)$) in constant time. 
   \item[(iv)] The computation of $\delta(\cal A)$ needs to determine the distances of all the red and blue points from ${\cal C}_{in}{(\cal A)}$ and ${\cal C}_{out}{(\cal A)}$. Hence, it requires $O(m+n)$ time. In other words, the computation of the small distance $\epsilon$ to shift the center of ${\cal A}$ needs $O(m+n)$ time.
   \item[(v)] Counting the number of blue points lying inside $\cal A$ or $\cal A'$ needs $O(m)$ time.
 \end{itemize}
The algorithm needs $O(n\log n)$ preprocessing time. The number of vertices in $FVD(R)$ as well as $VD(R)$ are $O(n)$. Case~(A.1), needs $O(n(m+\log n))$ time to choose ${\cal A}_{\min}$. 
In Case~(A.2), $FVD(R)$ is updated $m$ times to compute $FVD_i$ for each $m$ different blue points $p_i\in B$. This case needs $O(mn(m+n+\log n))$ time, i.e., $O(mn(m+n))$.
In Case~(A.3), we  compute $m\choose 2$ number of $FVD_{i,j}$ by updating $FVD(R)$. Hence, Case~(A.3) requires $O(m^2n(m+n))$ time.

Similarly, Case~(B.1) needs $O(n(m+\log n))$ time.
Case~(B.2) and Case~(B.3) need the same time as that of Case~(A.2) and Case~(A.3), respectively. 

Case (C.1) needs $O(n^2(m+n))$ time to compute ${\cal A}_{\min}$ since there are $O(n^2)$ numbers of possible points of intersections of the edges of $VD(R)$ with that of $FVD(R)$.
Each of Case~(C.2) and Case~(C.3) needs $O(n^2m)$ time to compute ${\cal A}_{\min}$.
 In Case~(C.4), the number of centers of all possible annuli $\cal A$  whose outer (resp. inner) circle is defined by $p_i$ (resp. $p_j$), is at most $2n$ which is determined by the points of intersection of boundaries of $VC(i)$ and $FR(j)$. Hence, it needs $O(nm)$ time to count the number of blue points inside all possible annuli whose ${\cal C}_{out}(\cal A)$ and ${\cal C}_{in}(\cal A)$ are defined by the blue points $p_i$ and $p_j$, respectively. There are $m\choose 2$ pairs of blue points; hence, it needs $O(m^2n(m+n))$ time to compute the optimal annulus. 
 Hence, the overall time complexity of the algorithm is $O(m^2n(m+n))$.\\ 
We can reuse the same storage space to store $FVD_i$ (resp. $VD_i$) for different values of $i$. It is also true for $FVD_{i,j}$ (resp. $VD_{i,j}$). Hence, it needs linear space. 
\end{proof}

\subsection{Generalized Red-Blue Circular Annulus Cover problem}
\label{joco8}
 In the generalized version, all the given bichromatic points in ${\mathbb R}^2$ are associated with different penalties. We compute an annulus $\cal A$ that has Property~\ref{default}, and we generate them with all possible color combinations of their defining points as stated in each of the different cases in the algorithm for $RBCAC$ problem (discussed in Section~\ref{joco7}). However, note that to have a minimum penalty, all the red points need not lie inside $\cal A$. For this reason, we need to generate all possible annuli exhaustively, except the case where all the defining points of ${\cal C}_{out}(\cal A)$ or ${\cal C}_{in}(\cal A)$ are blue (see Observation~\ref{single_red}). 

\begin{lema}
\label{circle_penalty_1}
    If one of the two circles ${\cal C}_{out}(\cal A)$ and ${\cal C}_{in}(\cal A)$, is defined by three points, then we can compute the annulus with a minimum penalty in $O(n(m+n)^3\log (m+n))$ time and $O(m+n)$ space. 
\end{lema}
\begin{proof}
    If ${\cal C}_{out}(\cal A)$ is defined by three points, then we can compute all possible such outer circles in $O(n(m+n)^2)$ ways since at least one of the defining points of ${\cal C}_{out}(\cal A)$ must be red. For each center $c$ of ${\cal C}_{out}(\cal A)$, we sort all the points in $R\cup B$ with respect to their distances from $c$ in $O((m+n)\log (m+n))$ time, and then compute ${\cal C}_{in}(\cal A)$ so that the annulus $\cal A$ has the minimum penalty. Computation of such ${\cal C}_{in}(\cal A)$ needs $O(m+n)$ time. Hence, we need  $O(n(m+n)^3\log (m+n))$ time. 

Similarly, if ${\cal C}_{in}(\cal A)$ is defined by three points, we need the same amount of time to generate the annulus with the minimum penalty. We need linear space to store all the points. Hence the result follows. 
\end{proof}

\begin{lema}
\label{circle_penalty_2}
    If both the circles ${\cal C}_{out}(\cal A)$ and ${\cal C}_{in}(\cal A)$, are defined by two points, then we can compute the annulus with a minimum penalty in $O(n^2(m+n)^3)$ time and $O(m+n)$ space. 
\end{lema}   
\begin{proof}
 Each of the two circles ${\cal C}_{out}(\cal A)$ and ${\cal C}_{in}(\cal A)$ have at least one red point that defines it, so we have $O(n^2(m+n)^2)$ possible positions of the center of the annulus ${\cal A}$. For each such annulus $\cal A$, we compute the points ($\in$ $R \cup B$) lying inside $\cal A$ and its penalty in $O(m+n)$ time. Hence, our result is proved. 
\end{proof}

\noindent We obtain the following result by combining Lemma~\ref{circle_penalty_1} and Lemma~\ref{circle_penalty_2}.
\begin{theo}
    The optimal solution to the Generalized Red-Blue circular annulus problem ($GRBCAC$) can be computed in $O(n^2(m+n)^3)$ time and linear space.
\end{theo}
\section{Concluding Remarks}
We have studied various rectangular and circular annulus cover problems for a given bichromatic point-set in one and two dimensions. Each red point in the bichromatic point-set has a non-covering penalty, whereas the blue point has a covering penalty. We design polynomial-time algorithms for each variation. 
Improving the running time of the generalized annulus cover problems in the future is interesting. 
Also, designing all these algorithms in higher dimensions remains challenging as a future work.
\bibliographystyle{spmpsci}
\bibliography{biblio}
\end{document}